\newif\ifanon\anonfalse
\newif\ifnotes\notesfalse

\documentclass[11pt,letter]{article}

\usepackage{breakcites}
\usepackage[utf8]{inputenc}
\usepackage[T1]{fontenc}

\usepackage{beton}
\usepackage{fullpage}

\usepackage{tikz,pgfplots}
\usepackage[margin=1in]{geometry}
\usetikzlibrary{calc}
\usepackage{caption}
\usepackage{amsmath}
\usepackage{amsfonts}
\usepackage{amssymb}
\usepackage{mathtools, xparse}
\usepackage[dvipsnames]{xcolor}
\usepackage{graphicx}
\usepackage{rotating}
\usepackage{braket}
\usepackage{url}
\usepackage{bm}
\usepackage{float}
\usepackage{soul}
\usepackage{multirow}
\definecolor{DarkBlue}{RGB}{0,0,150}
\definecolor{NotSoDarkBlue}{RGB}{15,15,210}
\definecolor{DarkRed}{RGB}{150,0,0}
\definecolor{DarkGreen}{RGB}{0,100,0}
\usepackage[pdfstartview=FitH,pdfpagemode=UseNone,colorlinks,linkcolor=DarkRed,filecolor=blue,citecolor=DarkRed,urlcolor=DarkRed,pagebackref,breaklinks]{hyperref}
\usepackage[ruled, algosection]{algorithm2e}
\usepackage{booktabs}
\usepackage[table]{xcolor}
\usepackage{adjustbox}

\usepackage{microtype}      
\usepackage{libertine}
\usepackage{libertinust1math}

\usepackage{orcidlink}

\newcommand{\commC}{\mathsf{CC}}

\newcommand{\ZZ}{\mathbb{Z}}

\newcommand{\F}{\mathbb{F}}
\newcommand{\N}{\mathbb{N}}
\newcommand{\poly}{\mathsf{poly}}
\newcommand{\polylog}{\mathsf{polylog}}

\allowdisplaybreaks

\DeclarePairedDelimiterX{\bigket}[1]{\bigg\lvert}{\bigg\rangle}{\,#1}

\usepackage{bbm}

\usepackage{framed}
\usepackage{multicol}
\usepackage[capitalise]{cleveref}
\usepackage{enumitem}

\usepackage{amsthm}
\theoremstyle{plain}
\newtheorem{theorem}{Theorem}[section]{\bfseries}{\itshape}
\newtheorem{informal-theorem}[theorem]{Informal Theorem}{\bfseries}{\itshape}
{\bfseries}{}
\newtheorem{lemma}[theorem]{Lemma}{\bfseries}{\itshape}
{\bfseries}{\itshape}
\makeatletter
\def\th@remarkit{\th@remark\itshape}
\makeatother
\theoremstyle{remarkit}
\newtheorem{remark}[theorem]{Remark}{\itshape}{}
\theoremstyle{plain}
{\bfseries}{}
{\bfseries}{}
\newtheorem{corollary}[theorem]{Corollary}{\bfseries}{\itshape}
\newtheorem{claim}[theorem]{Claim}{\bfseries}{\itshape}

\Crefname{prop}{Proposition}{Propositions}
\Crefname{fact}{Fact}{Facts}
\Crefname{game}{Game}{Games}
\Crefname{defn}{Definition}{Definitions}
\Crefname{remark}{Remark}{Remarks}
\Crefname{const}{Construction}{Constructions}

\usepackage{amsthm}
\theoremstyle{plain}

\newtheorem{definition}[theorem]{Definition}

\numberwithin{theorem}{section}
\numberwithin{conjecture}{section}
\numberwithin{problem}{section}

\AddToHook{env/lemma/begin}{\crefalias{theorem}{lemma}}
\AddToHook{env/corollary/begin}{\crefalias{theorem}{corollary}}
\AddToHook{env/definition/begin}{\crefalias{theorem}{definition}}
\AddToHook{env/proposition/begin}{\crefalias{theorem}{proposition}}

\AddToHook{env/remark/begin}{\crefalias{theorem}{remark}}
\AddToHook{env/claim/begin}{\crefalias{theorem}{claim}}
\crefname{claim}{Claim}{Claims}
\Crefname{claim}{Claim}{Claims}

\newcommand{\ord}[1]{\ensuremath{#1^{\mathrm{th}}}}

\newcommand{\ans}{\mathsf{ans}}
\newcommand{\Query}{\mathsf{Query}}

\newcommand{\Reconstruct}{\mathsf{Reconstruct}}

\newcommand{\supp}{\mathsf{supp}}

\newcommand{\query}{\mathsf{qu}}
\newcommand{\qu}{\query}

\newcommand{\Enc}{\mathsf{Enc}}

\newcommand{\Z}{\mathbb{Z}}

\newcommand{\SAns}{\mathsf{Answer}}

\newcommand{\Answer}{\SAns}

\newcommand{\DB}{\mathsf{DB}}
\newcommand{\db}{\mathsf{DB}}

\newcommand{\rgets}{\mathrel{\mathpalette\rgetscmd\relax}}
\newcommand{\rgetscmd}{\ooalign{$\leftarrow$\cr
    \hidewidth\raisebox{1.2\height}{\scalebox{0.5}{\ \rm R}}\hidewidth\cr}}%

\renewcommand{\vec}[1]{\mathbf{#1}}

 \newcommand{\uv}{\mathbf{u}}
\newcommand{\vecv}{\mathbf{v}} \newcommand{\wv}{\mathbf{w}}
 
\newcommand{\zv}{\mathbf{z}}

\newcommand{\one}{\mathbbm{1}}

\newcommand{\itpara}[1]{\medskip\noindent\textit{#1}}

\renewcommand{\paragraph}[1]{\medskip\noindent\textbf{#1}}

\newcommand{\calO}{\ensuremath{\mathcal{O}}}

\newcommand{\calR}{\ensuremath{\mathcal{R}}}

\newcommand{\boldtau}{\ensuremath{{\boldsymbol{\tau}}}}

\newcommand{\CRT}{\mathsf{CRT}}

\newcommand{\zo}{\ensuremath{\{0,1\}}} 

\newcommand{\TC}{\mathsf{TC}}
\newcommand{\NL}{\mathsf{NL}}
\newcommand{\NC}{\mathsf{NC}}
\newcommand{\CL}{\mathsf{CL}}
\newcommand{\Poly}{\mathsf{P}}

\newcommand{\answer}{\mathsf{ans}}

\newcommand{\ctrl}{\mathsf{ctrl}}

\newcommand{\ndb}{n_\mathsf{DB}}

\newcommand{\tep}{$\mathsf{TreeEval}$\xspace}
\newcommand{\tephl}{$\mathsf{TreeEval}_{h,\ell}$\xspace}
\newcommand{\tephlr}{$\mathsf{TreeEval}_{h,\ell, r}$\xspace}
\newcommand{\olg}{\mathsf{OneLevelGadget}}
\newcommand{\tephlmathmode}{\mathsf{TreeEval}_{h, \ell}}

\newcommand{\CT}[3]{\mathsf{CatSpaceTime}\left[#1,#2,#3\right]}

\newcommand{\eps}{\varepsilon}

\newcommand{\la}{\leftarrow}

\newcommand{\enc}{\mathsf{enc}}
\newcommand{\reg}{\mathsf{reg}}

\newcommand{\classL}{$\mathsf{L}$\xspace}
\newcommand{\Time}{\mathsf{Time}}
\newcommand{\Space}{\mathsf{Space}}

\newcommand{\ts}{\texorpdfstring}

\title{Catalytic Tree Evaluation from Matching Vectors}

\ifanon
    \author{}
    \date{}
\else
    \author{Alexandra Henzinger\thanks{Email: \texttt{amh@cs.stanford.edu}. Supported by the NSF Graduate Research Fellowship under Grant No. 2141064, a MIT EECS Great Educators fellowship, gifts from Apple, Google, and Meta, and NSF CNS-2054869.}\\MIT, Stanford
    \and
    Edward Pyne\thanks{Email: \texttt{epyne@mit.edu}. Supported by the NSF Graduate Research Fellowship.}\\MIT
    \and
    Seyoon Ragavan\thanks{Email: \texttt{sragavan@mit.edu}. Supported by Amazon Research Awards and NSF CNS-2534400, and in part by Jane Street. SR's access to ChatGPT Pro was supported by the UK AISI alignment project and the Simons Institute for the Theory of Computing, and his access to Claude Max was supported by Jane Street.}\\MIT}
    \date{\today}
\fi 

\begin{document}
\maketitle

\begin{abstract}
    We give new algorithms for tree evaluation (S. Cook et al. TOCT 2012) in the catalytic-computing model (Buhrman et al. STOC 2014). Two existing approaches aim to solve tree evaluation (\tep) in low space: on the one hand, J. Cook and Mertz (STOC 2024) give an algorithm for \tep running in super-logarithmic space $O(\log n\log\log n)$ and super-polynomial time $n^{O(\log\log n)}$. On the other hand, an earlier algorithm of J. Cook and Mertz (STOC 2020), run on a catalytic tape, solves \tep in logarithmic $O(\log n)$ free space and polynomial time, but with polynomial $O(\sqrt{n})$ catalytic space.

    We show that the latter result can be improved. We give a catalytic algorithm for \tep with logarithmic $O(\log n)$ free space, polynomial runtime, and subpolynomial $2^{O(\log^\eps n)}$ catalytic space (for any $\eps > 0$). Our result also implies an improved simulation of time by \textit{catalytic} space, via the reduction of Williams (STOC 2025).
    
    Our \tep algorithm comes from a framework, inspired by private information retrieval, that generically reduces tree evaluation to two combinatorial objects: matching-vector families and \emph{decoders} (a new object we define that closely resembles decoding polynomials). These ingredients' parameters directly govern an algorithm's resources: the vectors' dimension determines its catalytic space, and the decoder's sparsity determines its free space and runtime. In addition to enabling our new algorithm, this framework explains existing tradeoffs, for example why the Cook--Mertz algorithm needs super-logarithmic space but no catalytic tape. It also lets better matching-vector families or decoders plug in directly to give better algorithms.
\end{abstract}

\thispagestyle{empty}
\newpage
\thispagestyle{empty}
\tableofcontents
\newpage 
\pagenumbering{arabic}

\section{Introduction}
\label{sec:intro}

Space, unlike time, is a reusable resource. 
{\em Catalytic computing} studies the task of computing given a large amount of full memory, which may be modified but must appear unchanged at the end of the computation. 
A series of breakthrough results, initiated by Buhrman, Cleve, Kouck\'y, Loff, and Speelman~\cite{BCKLS14}, show that access to such full, {\em catalytic} space enables using less {\em free} space (ordinary space that can be overwritten and erased) to solve problems than otherwise known. For example, Buhrman et al.~\cite{BCKLS14} evaluate any uniform $\TC^1$ circuit with $O(\log n)$ free space and $\poly(n)$ catalytic space, while this task---which contains a complete problem for non-deterministic logspace ($\NL$)---is not known to be solvable with $o(\log^{2} n)$ free space alone. 

The power of catalytic space remains wide open. For instance, letting catalytic logspace ($\CL$) be the set of problems solvable in $O(\log n)$ free space and $\poly(n)$ catalytic space, it would be consistent with current knowledge both that all polynomial-time computations can be solved in $\CL$ (i.e., $\mathsf{P} \subseteq \CL$) or that $\CL$ is no stronger than the class of computations done in $O(\log^2 n)$ parallel time  (i.e., $\CL\subseteq \NC^2$).
One critical question towards understanding this power is: how small can catalytic space be for it to be useful? 

\medskip

In this paper, we draw a new connection between catalytic computing and cryptography. Building on this connection,  we obtain a new catalytic algorithm for the tree evaluation problem, abbreviated \tep~\cite{CMWBS12}. 
\tep has become a central protagonist in  complexity theory for its roles in\begin{itemize}[leftmargin=*]
    \item aiming to separate polynomial time $\mathsf{P}=\Time(\poly(n))$ and logarithmic space $\mathsf{L}=\Space(O(\log n))$~\cite{CM23} and
    \item giving new tradeoffs between space and time, namely that $\Time(T(n)) \subseteq \Space(\sqrt{T(n) \log T(n)})$~\cite{Wil25}.
\end{itemize} 
This work shows that  {\em logarithmic} free space and {\em subpolynomial} catalytic space suffice to solve \tep:

\begin{theorem}\label{thm:intro}
    For every $\epsilon > 0$, \tep can be solved in $O(\log n)$ free space, $2^{O(\log^{\epsilon} n)}$ catalytic space, and $\poly(n)$ runtime. (See~\cref{cor:eps} for a parameterized statement.)
\end{theorem}

 By contrast, the one-hot algorithm of Cook and Mertz~\cite[Theorem 1]{CM20}, run on a catalytic tape, solves \tep with superpolynomially more catalytic space: it uses $O(\log n)$ free space, $O(\sqrt{n})$ catalytic space, and polynomial runtime.\footnote{Evaluating a circuit for \tep with the catalytic algorithm of Buhrman et al.~\cite{BCKLS14} gives a similar result, but with $\poly(n)$ catalytic space.} At the same time, a later algorithm of Cook and Mertz~\cite{CM23} solves \tep with no catalytic space, but  with more free space and superpolynomially larger runtime than \cref{thm:intro}: it requires $O(\log n \cdot \log\log n)$ space and $n^{O(\log \log n)}$ time. 

Our algorithm also admits smooth tradeoffs that shrink the catalytic space at the expense of increasing the free space and the runtime.
For example, it implies the following new result:
\begin{theorem}\label{thm:intro-balance}
    \tep can be solved in $O(\log n \sqrt{\log \log n} \log \log \log n)$ free space, $\exp(\exp(O(\sqrt{\log \log n})))$ catalytic space, and $n^{O(\sqrt{\log \log n})}$ time. (See~\cref{cor:sqrt} for a parameterized statement.)
\end{theorem}

\medskip 
Two direct implications follow from our new catalytic algorithm for tree evaluation:

\paragraph{1. New time-space tradeoffs in the catalytic-space model.} Following the seminal reduction of Williams~\cite{Wil25}, our result implies the following corollary:

    \begin{corollary}\label{cor:TSintro}
    Let $T = T(n) \geq n$. Every language decided by a multitape Turing machine in time $T$ can be decided in $O(\sqrt{T})$ free space, $2^{O(T^{\epsilon})}$ catalytic space, and $2^{O(\sqrt{T})}$ time, for every $\epsilon > 0$. (See~\cref{cor:TS}.)
    \end{corollary}

    Relative to Williams' result, this shrinks the free space by a factor of $O(\sqrt{\log T})$ and shrinks the runtime by superpolynomial factors---his simulation runs in time $2^{O(\sqrt{T\log T})}$, against our $2^{O(\sqrt{T})}$---at the expense of introducing a large catalytic tape. Separately, combining the reduction of Williams with the one-hot algorithm of~\cite{CM20} gives that $\Time(T)$ can be decided with $O(\sqrt{T})$ free space and $2^{O(\sqrt{T})}$ catalytic space. Compared to this result, we shrink the catalytic tape by a superpolynomial factor. We state analogous tradeoffs for catalytic circuit evaluation.
    
\paragraph{2. A modular framework for tree evaluation.} Our algorithm is an instance of a general framework that reduces tree evaluation to two combinatorial objects: matching-vector families and \emph{decoders}, a new object closely resembling the decoding polynomials used for matching-vector codes. The framework generalizes the breakthrough algorithm of Cook and Mertz~\cite{CM23}, and it lets us apply tools from coding theory and cryptography to questions in catalytic computing.

    In particular, the core machinery driving our new catalytic \tep algorithm is a family of {\em matching vectors}~\cite{DBLP:journals/combinatorica/Grolmusz00}; informally, this is a collection of vectors whose inner products over composite moduli fall in a restricted set. Matching vectors give rise to the best known locally-decodable codes in the low query-complexity regime~\cite{DGY11} and to the most communication-efficient information-theoretic private-information-retrieval (PIR) schemes in the few-server regime~\cite{DG16,GKS25}. We observe that an algorithm at the heart of Cook-Mertz's \tep mirrors Reed-Muller PIR; by extending this connection to other PIR schemes---including Efremenko's matching-vector protocols---we obtain new tradeoffs among time, space, and catalytic space for tree evaluation.

    In this framework, the parameters of the two ingredients directly determine the resources of the resulting algorithm. Each node of the tree is evaluated with a number of recursive calls proportional to the decoder's sparsity, and the recursion stores only a term index and a few residues modulo $m$ per level, so the sparsity (together with the modulus $m$) determines the free space and the running time. The algorithm's registers hold vectors from the matching-vector family, so their dimension determines the catalytic space. This refines the encoding-based view of Cook and Mertz~\cite{CM21}, in which an algorithm's cost is governed by the length of an encoding of values as vectors and by the degree of the node functions under that encoding; in our framework, the matching-vector dimension and the decoder sparsity play these two roles. \Cref{tab:framework} places three algorithms in this framework.

    Viewed this way, the two earlier algorithms of Cook and Mertz sit at opposite corners of a tradeoff. The one-hot algorithm of~\cite{CM20} pairs a decoder of constant sparsity with vectors of polynomial dimension, so it needs polynomial catalytic space. The algorithm of~\cite{CM23} uses vectors of low dimension, so it needs no catalytic tape, but its decoder has polylogarithmic sparsity over a polylogarithmic prime modulus, which costs it super-logarithmic space and super-polynomial time. Our algorithm uses matching vectors over a constant composite modulus, which give constant sparsity and subpolynomial dimension at the same time. The same framework also recovers the Stoeckl--Goldreich improvement~\cite{Stoeckl,Gol24TE} of the Cook--Mertz algorithm, by grouping levels of the tree into nodes of higher fan-in and decoding sums of more than two scores (\cref{sec:blocking}).

\begin{table}[t]
    \centering
    \small
    \renewcommand{\arraystretch}{1.2}
    \setlength{\tabcolsep}{4pt}
    \begin{tabular}{@{}lcccccc@{}}
        \toprule
        & Shared modulus & \begin{tabular}[b]{@{}c@{}}Matching-vector\\ dimension\end{tabular} & \begin{tabular}[b]{@{}c@{}}Decoder\\ sparsity\end{tabular} & Free space & Catalytic space & Time \\
        \midrule
        \cite[Thm.~1]{CM20} & $O(1)$ & $O(\sqrt{n})$ & $O(1)$ & $O(\log n)$ & $O(\sqrt{n})$ & $\poly(n)$ \\
        \cite{CM23} & prime, $\polylog(n)$ & $O(\frac{\log n}{\log\log n})$ & $\polylog(n)$ & $O(\log n \log\log n)$ & none & $n^{O(\log\log n)}$ \\
        This work & composite, $O(1)$ & $2^{O(\log^{\eps} n)}$ & $O(1)$ & $O(\log n)$ & $2^{O(\log^{\eps} n)}$ & $\poly(n)$ \\
        \bottomrule
    \end{tabular}
    \caption{Three tree-evaluation algorithms viewed through our framework, on inputs of length $n$. The row for~\cite{CM20} reads its one-hot algorithm in our terms (\cref{sec:overview:mv}).}
    \label{tab:framework}
\end{table}
    
    Our result is the first use of such techniques in catalytic computing.  We view this connection, and the framework it yields, as the main contribution of the paper, and we are optimistic for further applications and links between information-theoretic cryptography, coding theory, and catalytic computing. Because the framework is modular, better ingredients translate directly into better algorithms. For example, there is a large gap between known upper and lower bounds for matching-vector families, and shorter families would directly improve our algorithm---even up to the point of speeding up Cook-Mertz to polynomial time, with no loss in space. In the terms of~\cref{tab:framework}, such a family would keep our constant modulus and sparsity while bringing the dimension down to $O(\log n \log\log n)$, small enough to hold the registers in free space. 

\begin{theorem}[Informal]\label{thm:inf}
	Suppose that (sufficiently uniform) matching-vector families over a constant composite modulus exist, with dimension not ruled out by known lower bounds. Then, \tep can be solved in $O(\log n \log\log n)$ ordinary space and polynomial time, with no catalytic tape at all. (See~\Cref{cor:betterMV}.)
\end{theorem}

The same accounting also exposes this framework's limits as a route to putting \tep in logspace. For a constant modulus, the lower bound of Bhowmick, Dvir, and Lovett~\cite{DBLP:journals/siamcomp/0001DL14}, which the work of Gowers, Green, Manners, and Tao~\cite{gowers2024conjecture} makes unconditional, shows that the dimension assumed in~\cref{thm:inf} is optimal up to constant factors; so within this framework, $O(\log n\log\log n)$ is the best ordinary-space bound attainable with a constant modulus.

\subsection{Related Work}
\label{sec:intro:rel}
\label{sec:bg:rel}

\paragraph{Algorithms for tree evaluation.}
There has been extensive work on algorithms~\cite{Bar89,BC92,CMWBS12,CM20,CM21,CM22,CM23} and lower bounds~\cite{CMWBS12,Liu13,EMP18,IN19} for tree evaluation, as well as on evaluating classes of circuits in the catalytic model~\cite{BCKLS14,AM25,CP25}. Though prior algorithms for tree evaluation use the catalytic model as an intermediate tool to save space, we give the first {\em catalytic} algorithm for tree evaluation with subpolynomial catalytic space, showing that subpolynomial catalytic space suffices to use only logarithmic free space. Previously, the best catalytic algorithm for tree evaluation was the one-hot algorithm of~\cite[Theorem 1]{CM20}, run on a catalytic tape, which uses $O(\sqrt{n})$ catalytic space (the algorithm implicit in the work of~\cite{BCKLS14} uses $\poly(n)$ catalytic space). 

\paragraph{Time-space tradeoffs.}
The breakthrough result of Williams~\cite{Wil25} improved on the 50 year old result of Hopcroft, Paul, and Valiant~\cite{HPV77} that $\Time(t)\subseteq \Space(O(t/\log t))$.  
To the best of our knowledge, the relationship between time (and space) and {\em catalytic} space remains wide open, and we are the first to give results in this area. It would even be consistent with current knowledge that $\Poly\subseteq\CL$, meaning that  any size-$n$ circuit could be evaluated in  $O(\log n)$ free space and $\poly(n)$ catalytic space. 

Concurrently, a manuscript released by OpenAI~\cite{OpenAI26OneTape} reports a simulation of time-$T$ deterministic one-tape computations in $O(T^{2/5} \cdot \polylog(T))$ free space and no catalytic space, improving on the prior one-tape simulation in $O(\sqrt{T})$ free space due to Hopcroft and Ullman~\cite{hopcroft1968relations}. 
In particular, this implies that RAM-time-$T$ computation can be simulated in space $O(T^{4/5} \cdot \polylog(T))$, giving the first space-bounded simulation of RAM computation with space polynomially sublinear in $T$. 
Compared with the reduction of~\cite{Wil25}, this result does not support a general multi-tape Turing machine, does not proceed via black-box reduction to \tep, and does not directly give a smaller catalytic-space bound or a faster algorithm for \tep. 

\paragraph{Matching vectors.} Yekhanin~\cite{Y08} first used matching vectors to obtain new families of locally-decodable codes, starting a cascade of subsequent works~\cite{KY09,Efr09,E12,IS10,DGY11,Yek12,chee2013query}. Since then, matching-vector codes have  found applications in low-communication PIR~\cite{DG16,GKS25,GR26} and other cryptographic protocols like conditional disclosure of secrets and linear secret sharing~\cite{LBA25}. To date, the best known families of matching vectors are due to Grolmusz~\cite{DBLP:journals/combinatorica/Grolmusz00}, building on work by Barrington, Beigel, and Rudich~\cite{Barrington1994}.

\subsection{Open Questions}
\label{sec:open}

Our work leaves open a number of concrete questions:
\begin{itemize}
    \item \emph{The space complexity of \tep:} is \tep indeed in classical logspace (\classL)? Separately, is it in simultaneous polynomial time and $\log^cn$ space for some $c<2$?

    Building on this work, one approach to showing \tep $\in$ \classL would be to design a procedure for materializing matching vectors ``on the fly'' from a much more succinct representation stored in catalytic space, which could bring the space usage of our algorithm down, perhaps all the way to $O(\log n)$. 
    A second approach would be to design shorter matching-vector families; for constant moduli, however, this route stops at $O(\log n\log\log n)$ space, as discussed after~\cref{thm:inf}.

    \item Can our techniques be combined with the result of OpenAI~\cite{OpenAI26OneTape} to obtain a simulation of one-tape time $T$ in $O(T^{2/5})$ free space and small catalytic space? This question is particularly interesting because it would imply a simulation of RAM-time $T$ in free space $T^{4/5}$ and small catalytic space, since RAM machines can be simulated in the one-tape model with a quadratic slowdown. While the OpenAI result is not a direct reduction to \tep, it is possible that similar techniques could be applicable there.
\end{itemize}
A more open-ended direction is whether known lower bounds on families of PIR schemes or on locally-decodable codes translate to lower bounds on \tep. For example, prior works~\cite{DBLP:conf/crypto/BeimelIM00} show that, in PIR, the  servers must inherently run in $\Omega(n)$ time to answer PIR queries. It remains open to prove an analogous bound on the runtime of \tep---or, alternatively, to use preprocessing techniques from the PIR literature for circumventing this bound to speed up \tep~\cite{HR26}.

Finally, we leave open whether it is possible to port other techniques from information-theoretic cryptography, and multi-party computation in particular, to obtain improvements in catalytic computing. Taking a step back, our methods leverage the following qualitative view: both catalytic computing and cryptographic protocols work by operating on {\em masked} values. In particular, catalytic computation requires intermediate results to be written onto a catalytic tape, meaning these results are stored and accessed while masked by the arbitrary contents of the catalytic tape $\tau$. Similarly, cryptography shows how to privately outsource a computation to untrusted parties by masking the computation's inputs with randomness. We are optimistic that decades of research in multi-party computation may broaden what is possible in the realm of catalytic algorithms.

\section{Technical Overview}

\subsection{Tree Evaluation and the Cook--Mertz Algorithm}

We begin with an overview of the Cook--Mertz algorithm for tree evaluation, following the presentation of Goldreich~\cite{Gol25}. 
We first formally define the tree evaluation problem:

\begin{definition}[\tep]\label{def:tep}
    For positive integers $h, \ell$, the problem \tephl is the following task~\cite{CMWBS12}: given a complete binary tree of height $h$, where
    \begin{itemize}[nosep]
        \item each leaf node is indexed by $u\in \zo^h$ and labeled by a value $v_u \in \zo^\ell$ and
        \item each internal node is indexed by $u\in \zo^{<h}$ and labeled by a function $f_u : \zo^\ell\times \zo^\ell \to \zo^\ell$,
    \end{itemize}
    evaluate the tree in a bottom-up manner by setting each internal node's value to be the evaluation of its function on its children's values, and output the value of the root node. In other words, output the value $v_\emptyset$ such that $v_u = f_u(v_{u 0}, v_{u1})$ for all $u \in \zo^{<h}$.
    The input length to the \tephl problem is thus $n = \Theta(2^h \cdot \ell \cdot 2^{2\ell})$, so we have $h, \ell \leq O(\log n)$. 
\end{definition}

\begin{definition}[Higher-arity \tep]\label{def:tephlr}
For an integer $r \geq 2$, we define \tephlr to be the analogous problem with a complete $r$-ary tree. Thus, \tephl is equivalent to $\mathsf{TreeEval}_{h, \ell, 2}$. 
\end{definition}

For clarity, we focus in this work on the arity $r=2$ case; \cref{lem:fanintwo} gives a simple reduction from \tep for any constant $r$ to \tep for $r=2$ that suffices for the application of~\cite{Wil25}.

\paragraph{A baseline algorithm.} A simple algorithm for \tep proceeds using a depth-first recursion to evaluate the value of each node, from leaves all the way up to the root. On input length $n \in \N$, tree height $h \in \N$, and value bitlength $\ell \in \N$, this solution runs in $\poly(n)$ time and uses space $O(h\ell) = O(\log^2 n)$. The bottleneck leading to this space cost is that, as the depth-first recursion proceeds all the way down the height-$h$ tree,  the algorithm must hold up to $O(h)$ internal nodes' values along a path from the root, each of which comprises $\ell$ bits. Indeed, \tep was posed in part as a candidate for separating the class \classL{} of all logarithmic space computations  from the class $\mathsf{P}$ of all polynomial-time computations  on exactly this intuition; it was considered plausible that roughly $h\ell= O(\log^2 n)$ bits of space are necessary, until the algorithms of Cook and Mertz~\cite{CM20,CM23} shattered this barrier.

\paragraph{The Cook--Mertz algorithm as a one-level gadget.}
The algorithm of Cook and Mertz~\cite{CM23} builds on this baseline $O(h\ell)$-space recursive algorithm yet shrinks its memory footprint by reusing the same  $o(h \ell)$ bits of space to hold a path of values from the root node. This seemingly-impossible task is accomplished using techniques from catalytic computation: instead of computing a new internal node value by writing it into fresh space, the algorithm computes the value by adding a multiple of it into an existing register. To be more precise, the improved algorithmic task at the core of Cook-Mertz~\cite{CM23}  can be distilled as the following subroutine, which we informally refer to as the ``one-level gadget'':

\begin{definition}[Informal definition of the one-level gadget]\label{def:olginformal}
    Let $m$ be a positive integer and $\calR$ a commutative ring containing $\ZZ_m$.
    For an element $x \in \calR$, we define the oracle $\calO_x(\tilde{\tau} \in \calR, \alpha \in \ZZ_m)$ to be that which updates a register $\tilde{\tau} \gets \tilde{\tau} + \alpha x$ in place.

    For a positive integer $\ell$ and a canonical encoding $\Enc: \zo^\ell \to \calR$---which, like every object in this paper, the algorithm must be able to generate entry by entry in small space (as made precise in \cref{rmk:uniformaccess})---the problem 
    $\olg_{\ell, \calR, \Enc}$ is the following task: given as input
    \begin{itemize}[nosep]
        \item three registers $\tau_0, \tau_1, \tau \in \calR$,
        \item oracle access to $\calO_{\Enc(x_0)}$ and $\calO_{\Enc(x_1)}$ for some unknown strings $x_0, x_1 \in \zo^\ell$,
        \item the truth table of a function $f: \zo^\ell \times \zo^\ell \to \zo^\ell$, and
        \item a scalar $\alpha^* \in \ZZ_m$,
    \end{itemize}
    perform the update $\tau \gets \tau + \alpha^* \cdot \Enc(f(x_0, x_1))$ while leaving the registers $\tau_0$ and $\tau_1$ unchanged at the end of the computation (though they may be temporarily modified throughout).
\end{definition}
Intuitively, the one-level gadget builds the oracle $\calO_{\Enc(f(x_0, x_1))}$ from both oracles~$\calO_{\Enc(x_0)}$ and $\calO_{\Enc(x_1)}$. Cook--Mertz then show that the gadget can be composed with itself recursively at low overhead in space. Specifically: if the $\olg$ algorithm makes $Q$ oracle queries, uses $O(\ell)$ bits of workspace between queries, of which only $\ell'$ bits are retained across queries, and runs in time $\poly(2^\ell)$, then the recursion keeps one $\ell'$-bit slot for each call on its call stack (at most one per level of the tree) plus a single shared $O(\ell)$-bit block  of scratch space. The resulting \tephl algorithm uses 
\begin{enumerate}[nosep]
    \item $O(\ell + h \cdot \ell')$ bits of {\em free} space, as well as
    \item the three aforementioned registers, which consist of $O(\log |\calR|)$ bits, persist through the whole computation, and are what the algorithm treats as its \emph{catalytic} space.
\end{enumerate}
The runtime is $\poly(Q^h, 2^\ell)$. (One should expect $\ell' \geq \log Q$, since the gadget must at least know which query it is up to.) \Cref{thm:algo,thm:TEPonelevel} give precise versions of these statements.

Looking ahead, our algorithm also follows this template: in~\Cref{sec:tep:step} we construct a new algorithm for the one-level gadget and in~\Cref{sec:gadgettofull} we lift this to an algorithm for \tep. Our innovation is in the one-level-gadget component, which we focus on for the remainder of the technical overview.

\paragraph{Constructing Cook--Mertz's one-level gadget.} 
Next, we outline Cook-Mertz's one-level gadget, using the presentation of Goldreich (and omitting some bookkeeping details to follow in the body of the paper)~\cite{Gol25}. 

First, some setup: let $\F$ be a finite field of prime order $p \in (2\ell + 2, 4\ell + 4)$. (One exists, by Bertrand's postulate.) Set $k := |\F| - 1 > 2\ell$, and let $\omega$ be a generator of $\F^*$, i.e., a primitive \ord{k} root of unity. Define the ring~$\calR = \F^\ell$ and the encoding function $\Enc:\zo^\ell \to \calR$ as $\Enc(x) = x$, viewing each entry of the string $x \in \zo^\ell$ as an element of the field~$\F$ in the natural way.

Given as inputs the three registers $\tau_0, \tau_1,$ and $\tau \in \calR$, the function $f: \zo^\ell \times \zo^\ell \to \zo^\ell$, and the scalar $\alpha^* \in \ZZ_m$, the Cook-Mertz gadget then proceeds as follows:
\begin{enumerate}[leftmargin=*]
    \item Let $\hat{f}: \F^\ell \times \F^\ell \to \F^\ell$ be the multilinear extension of the function~$f$. 
    Define the function~$g: \F \to \F^\ell$ by $g(x) := \hat{f}(x \tau_0 + x_0, x \tau_1 + x_1)$. 
    Note that $g$ has degree at most $2\ell < k$ and that the gadget's goal is to add the value~$\alpha^* \cdot g(0)$ into the register~$\tau$. 
    By Lagrange interpolation, it holds that:
    \begin{equation}\label{eq:lagrangeunivariate}
        g(0) = \frac{1}{k} \sum_{i = 0}^{k-1} g(\omega^i).
    \end{equation}
    \item Repeat for each $i = 0, 1, \ldots, k-1$:
    \begin{enumerate}
        \item Using queries to the oracles~$\calO_{\Enc(x_0)}$ and $\calO_{\Enc(x_1)}$, update $\tau_0 \gets \tau_0 + \omega^{-i} x_0$ and $\tau_1 \gets \tau_1 + \omega^{-i} x_1$.
        \item Update $\tau \gets \tau +\alpha^*/k \cdot \hat{f}(\omega^i \tau_0, \omega^i \tau_1)$.\footnote{The algorithm never writes down $\hat{f}$: the update is computed in one streaming pass over the truth table of $f$, evaluating the multilinear extension term-by-term and coordinate-by-coordinate and accumulating.}
        
        \item Using queries to the oracles~$\calO_{\Enc(x_0)}$ and $\calO_{\Enc(x_1)}$ again, restore $\tau_0 \gets \tau_0 - \omega^{-i} x_0$ and $\tau_1 \gets \tau_1 - \omega^{-i} x_1$.
    \end{enumerate}
\end{enumerate}
Correctness of this algorithm follows from equation~\eqref{eq:lagrangeunivariate}: the algorithm indeed performs the desired update $\tau \gets \tau + \alpha^* \cdot \Enc(f(x_0, x_1))$, while leaving the registers $\tau_0$ and $\tau_1$ unchanged at the end of the computation. 

In terms of efficiency, this gadget makes $4k = O(\ell)$ oracle queries; the only state it retains across queries is the counter $i$, the scalar $\alpha^*$, and its position in the loop: in sum, $\ell' = O(\log \ell)$ bits.
As a result, the corresponding \tephl algorithm uses free space $O(\ell + h \ell') = O(\ell + h\log \ell)$, catalytic space $O(\log|\calR|) = O(\ell \log \ell)$, and runtime $\poly(k^h, 2^\ell) = \poly(\ell^h, 2^\ell)$.
In this parameter regime, the ``catalytic'' registers occupy only $O(\ell \log \ell)$ bits so nothing is lost by folding them into the free space, as stated in~\Cref{cor:ordinary,cor:packedCM}. (We describe these registers as catalytic here to ease the reader into the presentation; in our main results, the catalytic space will become more important and super-polylogarithmic in size.)

\subsection{This Work: A Connection between \ts{\tep}{TreeEval} and Private Information Retrieval}\label{sec:connection}

Our improved algorithm for $\olg$ takes inspiration from cryptography. Its starting point is the observation that the univariate polynomial-interpolation trick in Cook--Mertz's one-level gadget appears almost verbatim in private information retrieval, as we describe next.  This view is key to our results, which leverage an arsenal of techniques from  private-information-retrieval to speed up $\olg$ and, as a result, \tep.

\paragraph{A primer on private information retrieval.} A private-information-retrieval (PIR) protocol~\cite{CGKS95,CGKS98,KO97} allows a user to read a bit from a remote database $\DB \in \zo^{\ndb}$, of which a copy is stored on $k \geq 2$ servers, without revealing to any of the servers what bit was read. In more detail, a PIR protocol consists of three algorithms:
\begin{enumerate}
    \item $\Query(i \in [\ndb]) \to  \query_0, \dots, \query_{k-1}$, which the user runs on the index $i  \in [\ndb]$ that it wants to read and which outputs $k$ PIR  queries, each of which is sent to one of the $k$ PIR servers.

    \item $\Answer(\DB \in \zo^{\ndb}, \query) \to \answer$, which each server runs on the database $\DB \in \zo^{\ndb}$ and on the PIR query it received and which outputs a PIR answer.

    \item $\Reconstruct(i \in [\ndb], \answer_0, \dots, \answer_{k-1}) \to \zo$, which the user runs on its index $i \in [\ndb]$ and on the list of $k$ PIR answers from all servers and which then outputs the recovered database bit.
\end{enumerate}

Taken together, these three algorithms satisfy two requirements: {\em correctness} and {\em privacy}. The correctness requirement guarantees that, if it was executed honestly, the user indeed recovers the \ord{i} bit of the database~$\DB$  at the end of the protocol. The  privacy requirement guarantees that the marginal distribution of each query $\query_j$, for $j \in [k]$, is independent of the index $i$ being queried.
In almost all information-theoretic PIR schemes~\cite{CGKS95,DBLP:conf/crypto/BeimelIM00,WY05,DBLP:journals/iacr/GhoshalLMDS24,HR26}, this is achieved by additively masking the index $i$, encoded as an element of some vector space, with uniformly-sampled randomness. 

In terms of PIR efficiency, two particular metrics are relevant to this work: (1)~the number of required servers $k \geq 2$, and (2)~the scheme's \emph{communication complexity} $\commC(\ndb)$, defined as the maximal bitlength of the $k$ PIR queries $\qu_1, \ldots, \qu_k$ and the $k$  PIR answers $\ans_1, \ldots, \ans_k$ on a worst-case, $\ndb$-bit database. 

\paragraph{Background: PIR from polynomial interpolation.} A classic ``Reed-Muller PIR'' scheme builds on multilinear extensions and univariate polynomial interpolation~\cite{CGKS95,DBLP:conf/crypto/BeimelIM00,WY05}, much like the Cook--Mertz gadget we just saw. In this PIR scheme, the user encodes its index~$i \in [\ndb]$ into a point in a vector space, then shoots a random line through this point, and sends each PIR server a point on this line. The $k$ PIR servers evaluate a multilinear extension $\hat f$, which encodes the database $\DB$, at the point they received and send the result back to the user. By design, these PIR answers trace out a low-degree univariate polynomial;
reconstruction interpolates that polynomial to recover the desired value $\DB_i$. 

Put in the same language we just used to describe Cook--Mertz, the steps in this PIR scheme are:
\begin{itemize}
    \item The  servers interpret the $\ndb$-bit database $\DB$ as a multilinear polynomial $\hat{f}: \zo^\ell \to \zo$, where $\ell = O(\log \ndb)$ and the polynomial's evaluations at a chosen set of points match the database values.
    
     Let $\F$ be a finite field of prime order in $(\ell + 2, 2\ell + 4)$, let $k := |\F| - 1 > \ell$, and let $\omega$ be a generator of $\F^*$, i.e., a primitive \ord{k} root of unity. 

    \item $\Query(i \in [\ndb])$: Parse the index~$i$ as a vector $x = x(i) \in \zo^\ell \in \F^\ell$. Sample a random mask $\tau \rgets \F^\ell$. For server $j \in [0, k-1]$, output $\qu_j = \tau + \omega^{-j} \cdot x$.

    Define the univariate polynomial~$g: \F \to \F$ to be the restriction of the database polynomial to the line spanned by the query points: i.e., $g(y) = \hat{f}(y\tau+x)$. The polynomial~$g$ has degree $\leq \ell$ and the user's goal is to learn $g(0)$, which is equal to the \ord{i} database bit.

    \item $\Answer(\DB, \qu)$: Server $j$ replies with  $\ans_j = \hat{f}(\omega^j \qu)$.

    \item $\Reconstruct(i, \ans_0, \ldots, \ans_{k-1})$: The user computes $\frac{1}{k} \sum_{j = 0}^{k-1} \ans_j$, which is $g(0) = \DB_i$ by equation~\eqref{eq:lagrangeunivariate}.
\end{itemize}

This view reveals a surprising connection: the Cook-Mertz $\olg$ for \tep~\cite{CM23} essentially evaluates the function $f$ by making white-box use of a $k$-server PIR protocol.

\paragraph{Connecting \tep to PIR.}
Recall the task in the $\olg$: the algorithm must evaluate
the function~$f(x_0, x_1)$---or, equivalently, retrieve the record at index $x_0 \| x_1$
from the truth table of $f$---without seeing either of the input strings $x_0$ or $x_1$ in
the clear. Its only access to these inputs is through the oracles
$\calO_{\Enc(x_0)}$ and $\calO_{\Enc(x_1)}$, which add multiples of
their encodings into the registers $\tau_0$ and $\tau_1$. After doing so, the values $x_0$ and $x_1$  appear masked by the registers' existing contents. 

This is the same challenge faced by a PIR protocol, which aims to read the bit at index $i$ from a database while exposing only additively masked copies of $i$.
With this in mind,
we set up an informal dictionary between the two settings in \cref{tab:pir-dictionary}. For Reed--Muller PIR, $c_j = \omega^{-j}$, $\Enc$ is the identity, and server~$j$ answers $\hat f(\omega^j \qu_j)$; the gadget column then becomes exactly the Cook--Mertz update $\tau \gets \tau + (\alpha^*/k)\, \hat f(\omega^j \tau_0, \omega^j \tau_1)$ above. This mapping from private information retrieval to catalytic computation of a function $f$ is  intended as an intuitive (rather than a formal) correspondence. Making it precise requires modifying the algorithm and imposing a number of structural requirements on the PIR scheme, which we describe informally next:

\begin{table}[t]
    \centering
    \renewcommand{\arraystretch}{1.25}
    \begin{tabular}{p{0.4\textwidth}@{\hskip 1.5em}p{0.53\textwidth}@{}}
        \toprule
        \multicolumn{1}{@{}l}{\textbf{$k$-server PIR}} & \multicolumn{1}{@{}l}{\textbf{One-level gadget}} \\
        \midrule
        \multicolumn{2}{@{}l}{\em Entities in protocol}\\[+4pt]
        Database $\DB$ & Truth table of the function~$f$ \\[+4pt]
        Index $i$ to read from $\db$& The input $(x_0, x_1)$ at which to evaluate $f$ \\[+4pt]
        PIR client and $k$ PIR servers & One algorithm, playing every role in sequence\\ 
        \midrule
        \multicolumn{2}{@{}l}{\em Computations}\\[+4pt]
        PIR query to server $j$: $\qu_j = \tau + c_j \cdot \Enc(i)$, for
        \begin{itemize}[leftmargin=*,nosep]
            \item a mask $\tau$ sampled uniformly at random,
            \item a server-dependent scalar $c_j$, and
            \item an encoding $\Enc(i)$ of the index~$i$.
        \end{itemize}  &
        Contents of register $\tau_\sigma$ ($\sigma \in \zo$) in round $j$, set by oracle calls: $\tau_\sigma + c_j \cdot \Enc(x_\sigma)$, for
        \begin{itemize}[leftmargin=*,nosep]
            \item the arbitrary initial register contents $\tau_\sigma$,
            \item a scalar $c_j$ chosen by the algorithm, and
            \item the encoding $\Enc(x_\sigma)$ of the input string $x_\sigma$. \vspace{+4pt}
        \end{itemize} \\
        PIR answer $\ans_j$, computed by server $j$ from $\DB$ and $\qu_j$ & One streaming pass over the truth table of $f$, computing the same function of the current contents of $\tau_0, \tau_1$ and adding it into register~$\tau$\\[+4pt]
        PIR reconstruction: user takes a fixed linear combination of the PIR answers & The same linear combination, accumulated in $\tau$ as the passes complete \\
        \midrule
        \multicolumn{2}{@{}l}{\em Cost metrics}\\[+4pt]
        Number of servers & Number of oracle calls \\[+4pt]
        Communication complexity $\commC(\ndb)$ & Catalytic space: the size of the registers $\tau_0, \tau_1, \tau$\\
        \bottomrule
    \end{tabular}
    \caption{An informal dictionary between $k$-server PIR, for schemes whose queries have the form ``mask plus a server-dependent multiple of an encoding of the index'', and the one-level gadget.}
    \label{tab:pir-dictionary}
\end{table}

\begin{enumerate}
    \item \emph{Composability.} PIR is a
    one-shot primitive: the user makes a PIR query, receives PIR answers, and recovers a database record, each of which may have its own type. The one-level gadget, on the other hand, is recursively composed with itself as part of the \tep algorithm. As a result, its output oracle needs to take exactly the form that its input oracles provide---an encoding added
    into a register---because the parent register $\tau$ plays the role of a child register one level up in the recursive \tep algorithm. Answers, in other words, must have the same type as
    queries. Reed--Muller PIR happens to satisfy this: the encoding
    is the identity, and queries and answers are both vectors in
    $\F^\ell$. For other PIR schemes it fails outright, and
    resolving it while also addressing the below bullet is one of our main challenges.

    \item \emph{Reconstruction is blind.} In PIR, the $\Reconstruct$ algorithm
    receives the index $i$ and all $k$ servers' answers at once. The $\olg$ is far more constrained: its input is stored in catalytic space (i.e., masked by the arbitrary register contents), and there is
    no free space available to store even a single  answer. Reconstruction must
    therefore  ``stream'' through the servers' answers one at a time, adding them into the same registers and  outputting the reconstructed database entry at the end (again in catalytic space, i.e., masked by arbitrary values). 
    
    One natural way of ensuring this is to require that the PIR scheme's reconstruction procedure applies a predetermined linear combination to each
    answer as it is produced.

    \item \em{The index arrives in two halves.} In PIR, the $\Query$ algorithm receives as input the index $i$. In the $\olg$, the inputs $x_0$ and $x_1$ are instead encoded by separate oracles into separate registers, so a PIR query
    for the pair $(x_0, x_1)$ must be assembled from independent PIR queries for each half (fetched by previous recursive calls, as per point (1)). 

    One way of overcoming this issue is by requiring  the PIR scheme's queries to be ``concatenation-friendly'': that is a PIR query for the pair $x_0 \| x_1$ is the concatenation of the PIR query for $x_0$ and $x_1$. As we will see in \cref{sec:tep}, while it is natural for the Reed-Muller PIR described above, massaging other PIR schemes to satisfy this property is more troublesome.

    \item \em{The catalytic tape is arbitrary.} PIR usually masks queries uniformly at random, while in $\olg$ we must work with a possibly adversarially chosen catalytic tape as the mask. These can be unified by simply restricting attention to PIR schemes that have perfect correctness. We will also make the structural requirement that the PIR queries be of the form ``mask plus a fixed encoding of the index.''
\end{enumerate}

\subsection{Adapting Matching-Vector PIR to \ts{\tep}{TreeEval}}
\label{sec:overview:mv}

Despite the aforementioned technical differences, the dictionary in~\Cref{tab:pir-dictionary} is a useful tool for anticipating what \tep efficiency a PIR protocol might offer. In the language of the introduction (\cref{tab:framework}), the number of servers becomes the decoder's sparsity, and the communication becomes the matching-vector dimension (times $\log m$). In particular,  the algorithm of Cook and Mertz~\cite{CM23} owes its low catalytic space to the low communication of Reed--Muller PIR, which achieves communication complexity $\tilde{O}(\log \ndb)$ when instantiated with $\tilde{O}(\log \ndb)$ servers. However, that same $O(\log \ndb)$-server parameter regime results in the Cook-Mertz algorithm requiring (1)~superpolynomial runtime and (2)~slightly superlogarithmic free space, since this super-constant number of servers corresponds to a super-constant number of recursive oracle calls at each internal node. The dictionary also explains the earlier one-hot algorithm of Cook and Mertz~\cite[Theorem 1]{CM20}: it corresponds to the four-server XOR (``cube'') scheme of~\cite{CGKS95}, whose constant number of servers yields $O(1)$ recursive calls per level, and hence $O(h+\ell)$ free space and polynomial time, but whose $\sqrt{\ndb}$ communication yields polynomial catalytic space $O(2^\ell)$.

\paragraph{Which PIR scheme to use?} With this in mind, the PIR literature points us to a different PIR protocol that offers hope for a new \tep algorithm, achieving a new tradeoff between runtime, free space, and catalytic space: {\em matching-vector
PIR}~\cite{Efr09,DG16,GKS25}. This family of PIR schemes, which gives rise to the most communication-efficient PIR known to date,  needs only a \emph{constant} number of servers to achieve already-{\em subpolynomial} communication: namely, 
communication $\exp(\tilde{O}((\log \ndb)^{1/t}))$ bits using $2^t$ servers, for any constant
$t \in \N$. 

Using matching-vector PIR, the dictionary in~\Cref{tab:pir-dictionary} thus predicts a \tep algorithm with free space
$O(h + \ell)$, polynomial runtime, and subpolynomial catalytic space $\exp(\tilde{O}(\ell^{1/t}))$. Making this scheme work is the goal of the remainder of this paper and what ultimately proves \cref{thm:intro}.

\paragraph{Efremenko's matching-vector PIR.} We begin with a review of matching-vector PIR.
Given a modulus $m$ (to be specific later on), let $\F$ be a finite field containing a primitive \ord{m} root of unity $\gamma$. Then, the scheme of~\cite{Efr09} is built from two related objects:

\begin{definition}[$S$-matching vectors]\label{def:smvfoverview}
    Let $S\subseteq \Z_m$ be a set containing $0$. The vectors $\uv_1, \ldots, \uv_N$ and $\vecv_1, \ldots, \vecv_N \in
    \Z_m^d$ form an \emph{$S$-matching-vector family} if their inner products are constrained in the following way:
    \begin{itemize}[nosep]
        \item $\langle \uv_i, \vecv_i \rangle = 0$ for every $i \in [N]$, and
        \item $\langle \uv_i, \vecv_j \rangle \in S \backslash \{0\}$ for every $i \neq j \in [N]$.
    \end{itemize}
\end{definition}

\begin{definition}[$A$-decoding polynomial]\label{def:decpolyoverview}
    Let $A \subseteq \Z_m$ be a set containing $0$. A polynomial
    $P \in \F[X]$ is an \emph{$A$-decoding polynomial} if its evaluations are constrained in the following way:
    \begin{itemize}[nosep]
        \item $P(1) = 1$, and 
        \item $P(\gamma^j) = 0$ for every $j \in A \backslash \{0\}$.
    \end{itemize}
    The \emph{sparsity} of $P$ is its number of monomials, i.e., the
    integer $k$ when we write $P(X) = \sum_{j=0}^{k-1} c_j X^{b_j}$ with
    distinct exponents and nonzero coefficients.
\end{definition}

For now, take $A = S$. Suppose a user wants to learn the $i$th entry of a database $\DB \in \F^{\ndb}$; then we can obtain Efremenko's scheme by working backwards from the goal:
\begin{align*}
    \DB_i &= \sum_{y \in [\ndb]} \DB_y \cdot \one[i = y] && \\
    &= \sum_{y \in [\ndb]} \DB_y \cdot \one\left[\langle \uv_i, \vecv_y \rangle = 0\right] && \text{(matching vector property)} \\
    &= \sum_{y \in [\ndb]} \DB_y \cdot P\left(\gamma^{\langle \uv_i, \vecv_y \rangle}\right) && \text{($P$ is $S$-decoding)} \\
    &= \sum_{y \in [\ndb]} \DB_y \sum_{j=0}^{k-1} c_j \, \gamma^{b_j \langle \uv_i, \vecv_y \rangle} && \text{(expand $P$)} \\
    &= \sum_{j=0}^{k-1} c_j \left(\sum_{y \in [\ndb]} \DB_y \cdot \gamma^{\langle b_j \uv_i, \vecv_y \rangle}\right). && \text{(swap the sums)}
\end{align*}
The PIR scheme works by indexing the servers by $j \in [k]$ and letting server $j$ be responsible for the corresponding inner sum over all $y \in [\ndb]$ in the above. While the user could recover $\DB_i$ by sending server $j$ the single vector $\qu_j = b_j \uv_i$, this would reveal the user's own index~$i$ and thus not be private. 

To ensure privacy, Efremenko simply adds a mask: the user samples
$\zv \rgets \Z_m^d$ uniformly at random and instead sends to server $j$ the query
$\qu_j = \zv + b_j \uv_i$, which is marginally uniform. Each summand above then picks up the factor
$\gamma^{\langle \zv, \vecv_y \rangle}$, and the same calculation
shows that the answers
$\ans_j = \sum_y \DB_y \cdot \gamma^{\langle \qu_j, \vecv_y \rangle}$
reconstruct the \emph{masked} record:
\begin{equation}\label{eq:mvmasked}
    \gamma^{\langle \zv, \vecv_i \rangle} \cdot \DB_i
    = \sum_{j=1}^{k} c_j \cdot \ans_j.
\end{equation}
As he knows $\zv$, the user computes $\langle \zv, \vecv_i \rangle$
and divides by the mask, thus recovering the desired entry $\DB_i$. 
In short, given an $S$-matching vector family of size $\ndb$ over $\Z_m^d$ and an $S$-decoding polynomial of sparsity $k$, the  $k$-server PIR protocol by~\cite{Efr09} achieves communication $O(d\log m)$ per server. 

Finally, to obtain the required mathematical objects, state-of-the-art constructions take $m$ to be a product of $t$ distinct, odd primes and set~$S := S^*_m$ to be the \emph{canonical set}, which comprises the $2^t$ residues that are congruent to $0$ or $1$ modulo each prime factor of
$m$. Then, Lagrange interpolation gives a decoding polynomial of degree $\leq |S|-1$ and therefore sparsity $k \leq |S| = 2^t$, and Grolmusz's
construction~\cite{DBLP:journals/combinatorica/Grolmusz00,Barrington1994}
gives an $S^*_m$-matching vector family of size $N$ with dimension
$d = \exp(\tilde{O}((\log N)^{1/t}))$. When $t$ and $m$ are constant, this yields PIR with a constant number of servers and
subpolynomial communication.

\paragraph{A one-level gadget from matching-vector PIR.}
We now translate this scheme to a $\olg$ for tree evaluation, proceeding with the informal dictionary of~\Cref{tab:pir-dictionary}. First, we index the
matching-vector family by $\zo^\ell$, so that each value
$x \in \zo^\ell$ has vectors $\uv_x, \vecv_x \in \Z_m^d$, and we take
$\calR = \Z_m^d$ with $\Enc(x) = \uv_x$; the oracle
$\calO_{\Enc(x_b)}$ then adds multiples of $\uv_{x_b}$ into the
register $\tau_b \in \Z_m^d$. For now, we regard the outputs of the function $f: \zo^\ell \times \zo^\ell \to \zo^\ell$
as elements of $\F$, with the parent register $\tau \in \F$.

Additionally, we assume that $S \cap (-S) = \{0\}$, so that a sum of
two permitted inner products vanishes only if both do; this holds for $S^*_m$ whenever $m$ is odd. Correspondingly, we take the polynomial~$P$ to be
$A$-decoding for the sumset $A = S + S := \{s + s' : s, s' \in S\}$. This then allows for the following computation:
\begin{align*}
    f(x_0, x_1) &= \sum_{y_0, y_1 \in \zo^\ell} f(y_0, y_1) \cdot \one[(y_0, y_1) = (x_0, x_1)] && \\
    &= \sum_{y_0, y_1} f(y_0, y_1) \cdot \one\left[\langle \uv_{x_0}, \vecv_{y_0} \rangle + \langle \uv_{x_1}, \vecv_{y_1} \rangle = 0\right] && \text{(matching vectors, $S \cap (-S) = \{0\}$)} \\
    &= \sum_{y_0, y_1} f(y_0, y_1) \cdot P\left(\gamma^{\langle \uv_{x_0}, \vecv_{y_0} \rangle + \langle \uv_{x_1}, \vecv_{y_1} \rangle}\right) && \text{($P$ is $(S{+}S)$-decoding)} \\
    &= \sum_{y_0, y_1} f(y_0, y_1) \sum_{j=1}^{k} c_j \, \gamma^{b_j \langle \uv_{x_0}, \vecv_{y_0} \rangle + b_j \langle \uv_{x_1}, \vecv_{y_1} \rangle} && \text{(expand $P$)} \\
    &= \sum_{j=1}^{k} c_j \left(\sum_{y_0, y_1} f(y_0, y_1) \, \gamma^{\langle b_j \uv_{x_0}, \vecv_{y_0} \rangle + \langle b_j \uv_{x_1}, \vecv_{y_1} \rangle}\right). && \text{(swap the sums)}
\end{align*}
Each inner sum can be computed in one streaming pass over the truth
table of the function~$f$, given the pair of queries $(b_j \uv_{x_0}, b_j \uv_{x_1})$. 

Finally, masking works much in the same way as in Efremenko's PIR: we do not have space to write down the queries $b_j \uv_{x_0}$ and $b_j \uv_{x_1}$, so instead we store them in catalytic registers, masking them with the pre-existing contents $(\tau_0, \tau_1)$ of the catalytic tape. Now, two oracle calls update
$\tau_0 \gets \tau_0 + b_j \uv_{x_0}$ and
$\tau_1 \gets \tau_1 + b_j \uv_{x_1}$.
Each summand picks up the factor
$\gamma^{\langle \tau_0, \vecv_{y_0} \rangle + \langle \tau_1, \vecv_{y_1} \rangle}$,
and the analogue of equation~\eqref{eq:mvmasked} reads
\begin{equation}\label{eq:gadgetmasked}
\begin{aligned}
    \gamma^{\langle \tau_0, \vecv_{x_0} \rangle + \langle \tau_1, \vecv_{x_1} \rangle} \cdot f(x_0, x_1)
    &= \sum_{j=1}^{k} c_j \cdot \ans_j, \\
    \text{where} \quad \ans_j
    &= \sum_{y_0, y_1} f(y_0, y_1) \, \gamma^{\langle \tau_0 + b_j \uv_{x_0}, \vecv_{y_0} \rangle + \langle \tau_1 + b_j \uv_{x_1}, \vecv_{y_1} \rangle},
\end{aligned}
\end{equation}
with the same coefficients $c_j$ as in equation~\eqref{eq:mvmasked}. Here, the coefficients $c_j$ are fixed in advance, so the reconstruction is indeed blind as we wanted in~\Cref{sec:connection}. To handle the correction factor of $\gamma^{\langle \tau_0, \vecv_{x_0} \rangle + \langle \tau_1, \vecv_{x_1} \rangle}$, the gadget separately computes this quantity using additional oracle calls and a standard catalytic algorithm for computing inner products (for
which the oracles must also provide the partner vectors
$\vecv_{x_b}$; see \cref{sec:tep:step}). In total, the gadget makes
$O(k)$ oracle calls and keeps $O(\log k)$ bits for bookkeeping, with
catalytic registers of $O(d \log m)$ bits. For Grolmusz's families over a constant number of primes, $k$ and $m$ are constants and $d = \exp(\tilde O(\ell^{1/t}))$, so once the gadget is made recursively composable (below), these costs yield exactly the algorithm the dictionary predicted: free space $O(h + \ell)$, polynomial time, and catalytic space $\exp(\tilde O(\ell^{1/t}))$, the last row of~\cref{tab:framework}.

\begin{remark}\label{rem:alternative}
    A natural alternative would avoid the assumption
$S \cap (-S) = \{0\}$ by decoding each child separately: apply an
$S$-decoding polynomial to each of the two indicators and multiply,
at the cost of $k^2$ streaming passes indexed by pairs of
monomials. This alternative also works for our main result; the benefit is that it only needs an $S$-decoding polynomial rather than an $(S+S)$-decoding polynomial, while the cost is that it requires more recursive calls---$O(k^2)$ to loop independently for each child, rather than the $O(k)$ we achieve with the sumset setup by looping ``in tandem''. These distinctions do not matter for our main result; however, the sumset route is better suited for recovering the $O(\log n \log \log n/\log \log \log n)$ space algorithm of Stoeckl~\cite{Stoeckl} and Goldreich~\cite{Gol24TE}, as we explain at the end of this section;~\cref{rmk:sumproduct} makes the comparison precise.
\end{remark}

\paragraph{Ensuring recursive composability.}
One issue remains: to compose the gadget recursively, the registers
$\tau_0$, $\tau_1$, and $\tau$ must all have the same type, since the
parent's register serves as a child's register one level up. This is
currently not the case, as $\tau_0, \tau_1 \in \Z_m^d$ while
$\tau \in \F$.

At first glance, this may seem like an innocent technicality, but this is actually a serious obstacle. To fix it, we would need a
commutative ring $\calR$ that simultaneously contains $\Z_m$ (i.e.,
$\calR$ has characteristic exactly $m$) and also has an $m$th root of
unity $\gamma$. This prevents $\calR$ from being a finite field;
however, one could hope to take $\calR$ to be a ring such as
$\Z_m[X]/(X^m - 1)$, the ring in which matching-vector codes are
usually presented~\cite{DGY11,Efr09}. However, it turns out that any
such ring \emph{cannot} have an $A$-decoding polynomial for our set
$A = S + S$, as the below lemma states. (Here, we are reinterpreting the definition of $A$-decoding polynomials to have outputs in $\calR$ in the natural way.)

\begin{lemma}\label{lem:nodecoder}
    Let $\calR$ be a commutative ring (with a multiplicative identity element) of characteristic
    exactly $m \geq 2$, and let $A \subseteq \Z_m$ be a set containing
    an element whose additive order in $\Z_m$ is a prime power $p^a$ with $a \geq 1$. Then,
    there do not exist $\gamma \in \calR$ and $P \in \calR[X]$
    satisfying
    \[
        \gamma^m = 1, \qquad P(1) = 1, \qquad
        P(\gamma^j) = 0 \quad \text{for every } j \in A \backslash \{0\}.
    \]
\end{lemma}

We defer the proof to \cref{app:nodecoder}. For $S = S^*_m$, the
hypothesis of the lemma is met: the residue that is $1$ modulo $p_1$
and $0$ modulo every other prime factor of $m$ lies in
$S^*_m \subseteq S + S$ and has additive order $p_1$, a prime. To bypass this issue, we work with an alternate
definition that gives us the same functionality as an $A$-decoding
polynomial but additionally has type-consistency baked into it:

\begin{definition}[$A$-decoder]\label{def:adecoderintro}
    Let $A \subseteq \Z_m$ be a set containing $0$. An
    \emph{$A$-decoder over $\Z_m$} of sparsity $k$ is a tuple of
    functions $\varphi_1, \ldots, \varphi_k: \Z_m \to \Z_m$ and elements
    $b_1, \ldots, b_k \in \Z_m$ such that:
    \begin{itemize}
        \item $\sum_{j=1}^{k} \varphi_j(z + a b_j) = 0$ for every $a \in A \backslash \{0\}$ and every $z \in \Z_m$, and
        \item $\sum_{j=1}^{k} \varphi_j(0) = 1$.
    \end{itemize}
\end{definition}

The two conditions above mirror the two defining properties of an
$A$-decoding polynomial, and setting $\varphi_j(y) = c_j \gamma^y$ recovers
essentially the guarantees of one---except that these $\varphi_j$ take values
in $\F$ rather than $\Z_m$. (\cref{lem:decoderequiv} records the exact correspondence with a weakened, ring-valued variant of the definition.) We construct $(S{+}S)$-decoders satisfying the full definition
over $\Z_m$, with sparsity $3^t$ for $S = S^*_m$, in
\cref{lem:compositedecoder}. With an $(S{+}S)$-decoder in place of $P$, the gadget
of the previous paragraph goes through essentially unchanged, now
with all three registers of type $\Z_m^d$. Instantiated with
Grolmusz's matching vectors, this yields a \tep algorithm with free
space $O(h + \ell) = O(\log n)$, catalytic space
$2^{O(\log^\epsilon n)}$, and $\poly(n)$ runtime (\cref{cor:eps}),
proving \cref{thm:intro}.

\paragraph{Alternative instantiations matching Cook--Mertz~\cite{CM23} and improvements~\cite{Stoeckl,Gol24TE}.}
Our framework has a second interesting regime: instead of Grolmusz's
matching vectors over a product of $t$ primes, we can use a single
odd prime $p$ together with a denser family of matching vectors whose inner products fall in a larger set $S$.
This recovers the guarantees of Cook--Mertz~\cite{CM23} and the $\log \log \log n$ factor improvement of Stoeckl~\cite{Stoeckl} and
Goldreich~\cite{Gol24TE}. We outline this construction below.

To define the matching vectors, recast strings $x \in \zo^\ell$ as
elements of $[M]^w$, for a coordinate bound $M$ and width $w$ to be
chosen shortly (the only constraint is that $M^w \geq 2^\ell$). Then define
\begin{equation}\label{eq:squarevectors}
    \uv_x = \left(1, x_1, \ldots, x_w, \sum_{j=1}^w x_j^2\right),
    \qquad
    \vecv_y = \left(\sum_{j=1}^w y_j^2, -2y_1, \ldots, -2y_w, 1\right),
\end{equation}
vectors of dimension $d = w + 2$, so that
\[
    \langle \uv_x, \vecv_y \rangle = \sum_{j=1}^{w} (x_j - y_j)^2.
\]
As an integer, this score is nonnegative, at most $D := wM^2$, and
zero exactly when $x = y$. 
Now choose the modulus $m$ to be a prime in $(2D, 4D)$ and
view the vectors in $\Z_m$. The inner products lie in the set
$S = \{0, 1, \ldots, D\} \subseteq \Z_m$, and since $m > 2D$, it follows that $S \cap (-S) = \{0\}$. 
Here, the decoding is done by the \emph{full-residue} decoder, whose up to $m$ terms annihilate every nonzero residue;~\cref{cor:packedCM} specifies how the interval structure of $S$ motivates this choice, as well as the resulting parameter regime.

Setting parameters, take $M = \ell$ and $w = \Theta(\ell / \log \ell)$.
Then $D = \poly(\ell)$, so $\log m = O(\log \ell)$, yielding catalytic space $O(d \log m) = O(\ell)$ bits.
The gadget makes $O(m) = \poly(\ell)$ oracle calls and keeps $O(\log \ell)$ bits for bookkeeping, yielding an algorithm for \tephl with space $O(\ell + h\log \ell) = O(\log n \log \log n)$ and running in time $\poly(\ell^h, 2^\ell) = n^{O(\log \log n)}$, recovering the result of~\cite{CM23} (made precise in~\cref{cor:packedCM}). Note that here, the catalytic space is so small that we are simply regarding it as ordinary free space, so this algorithm does not need to rely on the catalytic space model. In the terms of~\cref{tab:framework}, the low dimension $d = O(\ell/\log\ell)$ is why no catalytic tape is needed, while the decoder's $\poly(\ell)$ sparsity over a $\poly(\ell)$-size prime modulus costs $O(\log\ell)$ bits at each of the $h$ levels and $\poly(\ell)$ recursive calls per node, which is where the super-logarithmic space and super-polynomial time come from.

Finally, we can also adapt our instantiation to recover the improvements of Stoeckl~\cite{Stoeckl} and
Goldreich~\cite{Gol24TE} that achieved $O(\log n \log \log n/\log \log \log n)$ space. In the space bound above, the two costs are
unbalanced: the registers cost $O(\ell)$ bits once, while the
bookkeeping costs $O(\log \ell)$ bits at each of the $h$ levels. The idea of~\cite{Stoeckl,Gol24TE} is to rebalance the two by
increasing the branching factor of the \tephl instance: group $b$ consecutive levels of the
tree into a single node of fan-in $2^b$, so that the recursion depth
drops to $h/b$. Our framework is very compatible with this tweak: instead of requiring $S+S$ to not wrap around, we require $\underbrace{S+\ldots+S}_{2^b\text{ terms}}$ to not wrap around. So we take the prime $m$ to now be $O(2^bD)$ instead of just $O(D)$. As long as
$2^b \leq \ell$, say, we still have $\log m = O(\log \ell)$, so we still only have $O(\log \ell)$ bits of bookkeeping
space for each internal node, while we now need $O(2^b)$ catalytic registers of $O(\ell)$ bits each. The total space is thus
$O(2^b \ell + (h/b) \log \ell)$, up to lower-order terms, and
optimizing over $b$ recovers the
$O(\log n \log\log n / \log\log\log n)$ space bound
of~\cite{Stoeckl,Gol24TE}. We make this precise in
\cref{sec:blocking}.

Thinking back to~\Cref{rem:alternative}, we emphasize that this last optimization is where the sumset decoding
route becomes essential. Under the alternative route of multiplying two decoders, every child of an internal node would have its own independent monomial index, so an internal node would need $2^b \cdot O(\log \ell)$ bits of bookkeeping, as opposed to the $O(\log \ell)$ achieved with the sumset route. (\Cref{rmk:sumproduct} makes this precise.)

As summarized in~\cref{tab:framework}, better ingredients translate directly into better algorithms; \cref{cor:betterMV} quantifies what shorter matching-vector families would give.
\section{Preliminaries}

\paragraph{Notation.} For an integer $N$, we write $[N]$ to be the set $\{1, 2, \dots, N\}$.
We  use $\mathsf{reg} \gets \mathsf{state}$ to indicate that the register $\mathsf{reg}$ is updated to hold the contents of $\mathsf{state}$.
We  let $\exp(\cdot)$ and $\log(\cdot)$ denote the base-two exponential and logarithm, though the choice of base will not affect our results.

\subsection{The Catalytic-Space Model}
\label{sec:bg:cat}

We first define catalytic machines:
\begin{definition}
    A \textit{catalytic machine} $M$ is defined as a Turing machine in the usual sense---i.e.,~a read-only input tape, a write-only output tape, and a (space-bounded) read-write work tape---with an additional read-write tape known as the \textit{catalytic tape}. Unlike the ordinary work tape, the catalytic tape is initialized to hold an arbitrary string $\tau$, and $M$ has the restriction that for any initial setting of the catalytic tape, at the end of its computation the catalytic tape must be returned to the original state $\tau$. 
\end{definition}

\noindent We parameterize such a catalytic computation by three resources: time, space, and catalytic space.
\begin{definition}
    Let $\CT{C(n)}{S(n)}{T(n)}$ be the class of languages recognized by catalytic machines that, on inputs of size $n \in \N$, use $O(S(n))$ workspace and $O(C(n))$ catalytic space and run in time $O(T(n))$ in the worst case. 
\end{definition}

\begin{remark}[Representations of rings on the catalytic tape]\label{rmk:bitrep}
    Our algorithms (and those of Cook and Mertz and many others) interpret the catalytic tape as holding a vector in $\Z_m^{3d}$---the three registers of~\cref{thm:TEPonelevel}---though the definition of catalytic space gives a tape that consists of bits, with arbitrary initial contents that must be restored exactly. This translation requires care: for example, if we naively represent $\Z_3$ with $2$ bits, the tape could hold values that do not correspond to an element of the ring.
    A standard trick (cf.~\cite{CP25}) handles this at the cost of $B := \lceil \log(3dm) \rceil + 1$ bits per element and an additive $\poly(dm)$ term in the runtime. Partition the tape into $3d$ blocks of $B$ bits and let $K := \lfloor 2^B/m \rfloor$. Call a block holding the integer $z \in [0, 2^B)$ \emph{valid} if $z < Km$, and write a valid block as $z = qm + r$ with $0 \leq r < m$: the block represents the ring element $r$, and an update $r \gets r + c \bmod m$ leaves $q$ untouched, so it keeps the block valid and is invertible. Each block has fewer than $m$ invalid values, so fewer than $3dm < 2^B$ of the shifts $\Delta \in [0, 2^B)$ make some block $(z + \Delta) \bmod 2^B$ invalid; the algorithm finds a shift $\Delta$ making every block valid by exhaustive search in time $\poly(dm)$, applies it to every block, and stores $\Delta$ in $B = O(\log(dm))$ bits of workspace. It then runs with all register operations acting on the $r$-parts. Since the algorithm restores every register modulo $m$, each $r$-part returns to its initial value, and the frozen $q$-parts never change; subtracting $\Delta$ modulo $2^B$ from every block at the end therefore restores the tape to its initial bit contents exactly.
\end{remark}

\subsection{Existing Algorithms for Tree Evaluation}
\label{sec:bg:tep}

We recall prior algorithms for tree evaluation. The lowest-space procedure for tree evaluation is a brilliant algorithm due to Cook and Mertz:

\begin{theorem}[\cite{CM23}]
    $ \tephlmathmode \in \Space[O(\log(n)\cdot \log\log(n))]$. 
\end{theorem}
\noindent A further result of Stoeckl~\cite{Stoeckl}, written up by Goldreich~\cite{Gol24TE}, reduces the space by a $\log\log\log(n)$ factor.

Furthermore, the one-hot algorithm of Cook and Mertz~\cite[Theorem 1]{CM20} is a register program over $\F_2$ with three registers of $2^\ell$ bits each that works from arbitrary initial register contents and restores them; run on a catalytic tape, it gives:
\begin{theorem}[{\cite[Theorem 1]{CM20}}]
    $\tephlmathmode \in \CT{O(2^\ell)}{O(h+\ell)}{\poly(n)}$; in particular, \tep $\in \CT{O(\sqrt{n})}{O(\log n)}{\poly(n)}$.
\end{theorem}
\noindent Alternatively, the observation that \tephl can be computed by an unbounded fan-in circuit of depth $O(h)$ and size $\poly(n)$, combined with the catalytic circuit-evaluation algorithm of~\cite{BCKLS14}, gives the weaker bound $\tephlmathmode \in \CT{n^c}{O(\log n)}{n^c}$.

Finally, we can reduce from \tep with constant (but greater than $2$) fan-in to \tep with fan-in $2$, at a mild cost to the height. The instances of interest in~\cref{sec:timespace} are exponentially large and presented implicitly, so we state the reduction at the level of single bits: it converts access to individual entries of the $r$-ary instance into access to individual entries of a binary instance.
\begin{lemma}\label{lem:fanintwo}
    Let $r \geq 2$. There is an algorithm that, given access to individual truth-table entries and leaf values of a $\mathsf{TreeEval}_{h,\ell,r}$ instance $I$ (\cref{def:tephlr}), computes any requested bit of a $\mathsf{TreeEval}_{h\lceil \log r\rceil,\,\ell \cdot \lceil r/2\rceil,\,2}$ instance $I'$ whose root value is the root value of $I$ padded with zeros. It uses $O(h \log r +\ell r)$ bits of workspace beyond that access, and makes at most one access to $I$ per requested bit.
\end{lemma}
\begin{proof}
    We set $\ell'=\ell \cdot \lceil r/2\rceil$ and think of every value as an $\ell$-bit string padded to length $\ell'$. For every internal node $u$ of $I$ with children $u_1,\ldots,u_r$, we place a complete binary gadget tree of height $\lceil \log r\rceil$ whose leaves are the gadgets of $u_1,\ldots,u_r$ together with dummy slots, assigning the $r$ original children to the gadget's leaves by splitting them as evenly as possible at every level (the first $\lceil a/2 \rceil$ of $a$ children to the left subtree, the rest to the right). A dummy slot below a bottom-level node of $I$ is a leaf of $I'$ valued $0^{\ell'}$; a dummy slot elsewhere is filled with a complete subtree of gates that ignore their children and output $0^{\ell'}$, over $0^{\ell'}$-valued leaves, so that $I'$ is a complete binary tree of height $h\lceil \log r\rceil$. Every non-root node of the gadget collects the values passed from its children: a gadget node at distance $j \geq 1$ from the gadget root then holds the concatenation of the (unpadded) $\ell$-bit values of the at most $\lceil r/2^{j}\rceil \leq \lceil r/2 \rceil$ original children below it---iterated halving, since $\lceil \lceil a/2 \rceil /2\rceil = \lceil a/4 \rceil$---padded to length $\ell'$, which fits. (A naive left-to-right assignment would not fit: for $r = 5$, one child of the gadget root would hold four values.) The gadget root, whose collected inputs are $v_{u_1},\ldots,v_{u_r}$, has output $f_u(v_{u_1},\ldots,v_{u_r})$ padded by $0^{\ell'-\ell}$. This preserves the value at the root of $I$, up to the padding.

    For the bit access: from the index of a requested bit of $I'$, identify the original node $u$, the position within its gadget, and the position within an $\ell'$-bit value, by integer arithmetic on the index. A collection gate's entry is a fixed selection of the bits of its two arguments, computed directly from the index; a gadget-root entry requires reading one truth-table entry of $I$, at an argument tuple extracted from the index; a constant-$0^{\ell'}$ gate's entry is computed from the index alone; and a leaf of $I'$ is a leaf of $I$, padded, or $0^{\ell'}$. The index has length $O(h\log r + \ell r)$, and the arithmetic above needs no more workspace than that.
\end{proof}

\subsection{Matching Vector Families}
\label{sec:bg:mv}

We next define matching-vector families and verify  useful facts about them. Let $m > 1$ be a positive integer modulus.

\begin{definition}[Log-space uniformity]
    Recall the definition of $S$-matching-vector families in~\Cref{def:smvfoverview}. 
    We say such a (sequence of) families $\{\uv^{(N)}_i, \vecv^{(N)}_i: N \in \N, i \in [N]\}$ is log-space uniform if there is an algorithm that, on input $(1^N,i,j,m)$ where $i,j\in [N]$, prints $(\vec u_i,\vec v_j)$ to the output tape and runs in space $O(\log N+\log d+\log m)$.\footnote{Note that we allow space logarithmic in the modulus (versus doubly logarithmic). This does not affect the ultimate complexity of any of our constructions.}
\end{definition}

\begin{remark}[Space-$s$ uniformity and coordinate access]\label{rmk:uniformaccess}
    More generally, we say such a family is \emph{space-$s$ uniform} if the printing algorithm above runs in space $O(s)$; log-space uniformity is the case $s = \log N + \log d + \log m$. When a construction depends on auxiliary parameters (such as a list of prime factors of $m$, or the digit bound and width in~\cref{s-matching-CM}), the printing algorithm receives them as part of its input; each construction below states its parameters. Two standard observations about this interface will be used without further comment. First, an algorithm that needs a single coordinate $(\uv_i)_c$ or $(\vecv_j)_c$, rather than the whole vector, can run the printing algorithm, suppress all output before the $c$th coordinate, and stop after producing it; this costs $O(\log d)$ additional space for a counter. Second, a space-$O(s)$ algorithm runs in time $2^{O(s)}$ (assuming the input has length at most $2^{O(s)}$), which suffices for all of our running-time claims.
\end{remark}

We  describe two well-known instantiations of $S$-matching-vector families, which will be a crucial ingredient in our construction and in our unified presentation with the Cook-Mertz algorithm.

\itpara{Construction 1: An $S$-matching-vector family with a composite modulus $m$ and restricted set $S$.}

\begin{definition}[Canonical set]\label{def:canonicalset}
    Let $t \geq 1$ and suppose $p_1, \ldots, p_t$ are distinct primes and $m = p_1p_2\ldots p_t$.
    We say that the canonical set of $\ZZ_m$ is the set $S^*_m \subseteq \ZZ_m$ of residues that are either 0 or 1 mod $p_i$ for all $i \in [t]$; in particular $0 \in S^*_m$. The $2^t$ patterns of residues give distinct elements by the Chinese remainder theorem, so $|S^*_m| = 2^t$.
\end{definition}

\begin{theorem}[\cite{Barrington1994,DBLP:journals/combinatorica/Grolmusz00,DGY11}]\label{thm:mvgeneral}
    Let $m = p_1p_2\ldots p_t$ be a product of $t \geq 1$ distinct primes, given as input, let $S$ be the canonical set of $\ZZ_m$, and let $H \geq w \geq 1$ be integers.
    Then, there exists a space-$O\!\left(\log N + (mw)^{1/t}\log(H+1) + \log m\right)$-uniform $S$-matching-vector family over $\Z_m^d$ of size $N$, where:
    \begin{align*}
        N = \binom{H}{w} && \text{and} &&
        d \leq \sum_{j = 0}^{\lfloor (mw)^{1/t} \rfloor} \binom{H}{j}.
    \end{align*}
    In both parameter regimes of~\cref{cor:mvspecific}, the space bound is $O(\log N + \log m)$, so those families are log-space uniform; this is verified in~\cref{sec:mvspecificproof}.
\end{theorem}
\begin{proof}[Proof sketch]
    This follows from setting parameters suitably in~\cref{lem:dgyuniform}, our account of~\cite[Lemma 11]{DGY11}.
    For each $i \in [t]$, take $e_i$ to be minimal such that $p_i^{e_i} > (mw)^{1/t}/p_i$ (possibly $e_i = 0$, when $p_i > (mw)^{1/t}$).
    This implies that $\prod_i p_i^{e_i} > mw/\prod_i p_i = w$, and moreover for any $i$ we have $p_i^{e_i} \leq (mw)^{1/t}$, so the dimension $\sum_{j < \max_i p_i^{e_i}} \binom{H}{j}$ given by~\cref{lem:dgyuniform} is at most the bound displayed above, and the space bound $O(\log N + \delta \log H + \log m)$ of~\cref{lem:dgyuniform} is at most the one displayed above since $\delta \leq (mw)^{1/t}$ and $\log H \leq \log(H+1)$. While~\cite[Lemma 11]{DGY11} does not state any uniformity, it is immediate from their construction; we verify this in~\Cref{app:unif}.
\end{proof}
\noindent
The following corollary is straightforward; we defer its proof to~\cref{sec:mvspecificproof}.
\begin{corollary}\label{cor:mvspecific}
    All of the following hold:
    \begin{enumerate}
        \item\label{item:grolmusz} (\cite{DBLP:journals/combinatorica/Grolmusz00}) Let $t \geq 2$ be a constant, $p_1, \ldots, p_t$ be $t$ fixed odd primes, $m = p_1p_2\ldots p_t$, and $S$ be the canonical set of $\ZZ_m$.
        Then, for any $\ell \geq 1$, there exists a log-space-uniform $S$-matching-vector family over $\Z_m^d$ of size exactly $2^\ell$ with dimension $d = \exp\left(O_m\left(\ell^{1/t} (\log (\ell+2))^{1-1/t}\right)\right)$, where the implied constant depends on the fixed primes.

        \item\label{item:manyprimes}
        For any $\ell$ larger than an absolute constant, let $t = \left\lceil\sqrt{\log \ell}\right\rceil$, $p_1, \ldots, p_t$ be the first $t$ odd primes, $m = p_1p_2\ldots p_t$, and $S$ be the canonical set of $\ZZ_m$.
        Then, there exists a log-space-uniform $S$-matching-vector family over $\Z_{m}^d$ of size exactly $2^\ell$ with dimension $d = \exp(\exp(O(\sqrt{\log \ell})))$.
    \end{enumerate}
\end{corollary}

\itpara{Construction 2: An $S$-matching-vector family with an arbitrary, odd modulus $m$ and a larger set $S$.}

\begin{lemma}\label{s-matching-CM}
    Let $M, w$ be positive integers and let $S = \{0, 1, \ldots, wM^2\}$. Let $m > wM^2$ be an odd prime, with $M$, $w$, and $m$ given as input. Then there exists a space-$O(\log N + \log m)$ uniform $S$-matching-vector family over $\Z_m^{w+2}$ of size $N = M^w$.
\end{lemma}
\begin{proof}
    Index the elements of $[N]$ as vectors in $[M]^w$, identifying $i \in [N]$ with the base-$M$ digit vector of $i - 1$ (with digits shifted into $[M]$); both directions of this identification are computable in space $O(\log N)$ by standard arithmetic. For $x \in [M]^w$, define:
    \begin{align*}
        \uv_x &:= \left(1, x_1, \ldots, x_w, \sum_{c = 1}^w x_c^2\right),\text{ and}\\
        \vecv_x &:= \left(\sum_{c = 1}^w x_c^2, -2x_1, \ldots, -2x_w, 1\right).
    \end{align*}
    For any $x, y \in [M]^w$, we have $\langle \uv_x, \vecv_y \rangle = \sum_{c = 1}^w (x_c - y_c)^2 \in [0, wM^2]$, where the containment holds over the integers. Since $m > wM^2$, it follows that the inner product is 0 mod $m$ if and only if $x = y$, so this is an $S$-matching-vector family.

    For uniformity: on input $(1^N, i, j, m)$, the algorithm recovers the digit vectors $x$ of $i$ and $y$ of $j$ and prints the $w+2$ coordinates of $\uv_x$ and $\vecv_y$ in order. Each coordinate is $1$, a digit $x_c \in [M]$, a value $-2y_c$ reduced mod $m$, or the sum $\sum_c x_c^2$, which is a nonnegative integer smaller than $m$ and is computed by accumulation. All quantities involved are at most $m^2$, so the space used is $O(\log N + \log m)$.
\end{proof}
\begin{remark}\label{rmk:squarednoncancel}
    The score set $S$ above contains $0$, as our convention requires, and the hypothesis $m > wM^2$ makes the diagonal inner products the only ones that vanish. The noncancellation property $S \cap (-S) = \{0\}$ required by our main theorem is stronger: it holds if and only if $m > 2wM^2$. The instantiations in~\cref{sec:tep} choose the prime $m$ accordingly.
\end{remark}

\noindent The instantiations generate the prime $m$ themselves, using the following standard fact.
\begin{lemma}[Prime selection]\label{lem:primeselect}
    Given a positive integer $D$ in binary, the least prime in $(2D, 4D)$ exists and can be found in space $O(\log D)$.
\end{lemma}
\begin{proof}
    Existence follows from Bertrand's postulate (for every $n > 1$ there is a prime strictly between $n$ and $2n$), applied with $n = 2D$. The algorithm tests the integers in $(2D, 4D)$ in increasing order by trial division; every quantity involved is at most $4D$, so the algorithm runs in space $O(\log D)$ (and time $\poly(D)$).
\end{proof}

\subsection{\ts{$A$}{A}-Decoders}

Our decoders are applied to sums of scores, so we first fix notation for sumsets.

\begin{definition}[Sumsets and noncancellation]\label{def:sumset}
    For $S \subseteq \ZZ_m$ and an integer $r \geq 1$, the \emph{$r$-fold sumset} of $S$ is $rS := \{s_1 + \cdots + s_r : s_1, \ldots, s_r \in S\}$; we write $S + S := 2S$. We say $S$ is \emph{$r$-noncancelling} if $s_1 + \cdots + s_r = 0$ with $s_1, \ldots, s_r \in S$ implies $s_1 = \cdots = s_r = 0$. When $0 \in S$, being $2$-noncancelling is equivalent to $S \cap (-S) = \{0\}$.
\end{definition}

A second ingredient of our construction is a family of sparse $A$-decoders, which we defined in~\Cref{def:adecoderintro} and is heavily inspired by the $S$-decoding polynomials defined by Efremenko~\cite{Efr09}.
Our algorithm must evaluate decoder terms on demand, so we also fix an access interface and count its cost.

\begin{definition}[Uniform $A$-decoder]\label{def:uniformdecoder}
    Let $(\varphi_1, \ldots, \varphi_k; b_1, \ldots, b_k)$ be an $A$-decoder over $\ZZ_m$ (\cref{def:adecoderintro}), presented together with a parameter string; for the constructions below, the parameters are the prime factors $p_1, \ldots, p_t$ of $m$ (Construction 1) or $m$ itself (Construction 2). We say the decoder is \emph{space-$s$ uniform} if there is an algorithm that, given the parameters, an index $j \in [k]$, and $z \in \ZZ_m$, outputs the shift $b_j$ and the value $\varphi_j(z)$ in space $O(s)$.
\end{definition}

Our $A$-decoders are distilled from the following polynomial formulation, a ring-valued counterpart of the field-valued \cref{def:decpolyoverview}, which is how decoding is organized in the matching-vector-codes literature~\cite{Efr09,DGY11}. We call it \emph{weak} because its second condition asks only that the constant coefficient of $P(1)$ be $1$, rather than that $P(1) = 1$; over the ring $\calR$ below, the stronger requirement cannot be met once $A$ contains an element of prime-power additive order (\cref{lem:nodecoder}), as every set $A$ we use does. We record the formulation and then the exact correspondence.

\begin{definition}[Weak $A$-decoding polynomials]\label{def:weakdecpoly}
    Let $m > 1$ be a positive integer, let $A \subseteq \ZZ_m$ be a subset containing 0, and let $\calR = \ZZ_m[X]/(X^m-1)$.
    We say a polynomial $P \in \calR[Z]$ is a \emph{weak $A$-decoding polynomial} if:
    \begin{enumerate}
        \item for every $a \in A \setminus \{0\}$, we have $P(X^a) = 0$, and
        \item $P(X^0) = P(1)$ has constant coefficient equal to 1.
        (The \emph{constant coefficient} of an element of $\calR$ is the coefficient of $X^0$ in its canonical representative $\sum_{u=0}^{m-1} c_u X^u$.)
    \end{enumerate}
    The sparsity of $P$ is its number of non-zero coefficients (equivalently, its number of monomials).
\end{definition}

The definition of $A$-decoders may look ad hoc next to the polynomial formulation above, but the two notions contain exactly the same information: an $A$-decoder is a weak decoding polynomial with its coefficients extracted. The following lemma makes this precise.

\begin{lemma}[Equivalence of decoders and weak decoding polynomials]\label{lem:decoderequiv}
    Let $m > 1$, let $A \subseteq \ZZ_m$ be a set containing $0$, and let $\calR = \ZZ_m[X]/(X^m-1)$. For every $k \geq 1$, there exists an $A$-decoder over $\ZZ_m$ of sparsity at most $k$ if and only if there exists a weak $A$-decoding polynomial over $\calR$ of sparsity at most $k$. The translation preserves the shifts: the decoder $(\varphi_1, \ldots, \varphi_k; b_1, \ldots, b_k)$ and the polynomial $P(Z) = \sum_{j=1}^k C_j Z^{b_j}$ correspond via $C_j = \sum_{y \in \ZZ_m} \varphi_j(y) X^{-y}$, i.e., $\varphi_j(y)$ is the coefficient of $X^{-y}$ in $C_j$.
\end{lemma}
\begin{proof}
    First suppose $(\varphi_1, \ldots, \varphi_k; b_1, \ldots, b_k)$ is an $A$-decoder, and define $C_j$ and $P$ as in the statement, merging any terms with equal shifts (which only lowers the sparsity). For any $a \in \ZZ_m$, substituting $y = z + ab_j$ in each term gives
    \begin{equation}\label{eq:decoderequiv}
        P(X^a) = \sum_{j=1}^k \sum_{y \in \ZZ_m} \varphi_j(y)\, X^{ab_j - y} = \sum_{z \in \ZZ_m} \left(\sum_{j=1}^k \varphi_j(z + ab_j)\right) X^{-z}.
    \end{equation}
    For $a \in A \setminus \{0\}$, every inner sum on the right vanishes by the first decoder condition, so $P(X^a) = 0$. For $a = 0$, the constant coefficient of $P(1)$ is the $z = 0$ inner sum, namely $\sum_j \varphi_j(0) = 1$ by the second condition. Hence $P$ is a weak $A$-decoding polynomial of sparsity at most $k$.

    Conversely, given a weak $A$-decoding polynomial $P(Z) = \sum_{j=1}^k C_j Z^{b_j}$ with distinct shifts $b_j$, write $C_j = \sum_{u \in \ZZ_m} c_{j,u} X^u$ and define $\varphi_j(y) := c_{j,-y}$, so that~\eqref{eq:decoderequiv} again holds: the coefficient of $X^{-z}$ in $P(X^a)$ is $\sum_j \varphi_j(z + ab_j)$. Then $P(X^a) = 0$ for all $a \in A \setminus \{0\}$ gives the first decoder condition for every $z \in \ZZ_m$, and the constant coefficient of $P(1)$ being $1$ gives $\sum_j \varphi_j(0) = 1$.
\end{proof}

We next give two constructions of $A$-decoders, corresponding to the parameter regimes for the $S$-matching-vector family constructions given above. Although our definition of $A$-decoders is new (to the best of our knowledge), these decoders are assembled from standard ingredients (cf.~\cite{Efr09}).

\itpara{Construction 1: An $A$-decoder for composite $m$ and the sumset of the canonical set.}

\begin{lemma}\label{lem:compositedecoder}
    Let $p_1, \ldots, p_t$ be $t$ distinct odd primes, given as input, and let $m = p_1p_2\ldots p_t$. Let $A = S^*_m + S^*_m$ (\cref{def:sumset}); equivalently, $A$ is the set of residues in $\ZZ_m$ that are in $\{0, 1, 2\} \bmod{p_i}$ for all $i \in [t]$.
    Then, there exists a space-$O(\log m)$ uniform $A$-decoder over $\ZZ_m$ with sparsity $\leq 3^t$, whose term-access algorithm moreover runs in time $\poly(\log m)$.
\end{lemma}
\begin{proof}
    We first check that the two descriptions of $A$ agree. An element of $S^*_m + S^*_m$ is, modulo each $p_i$, a sum of two residues in $\{0,1\}$, hence lies in $\{0,1,2\} \bmod p_i$. Conversely, given $a$ in $\{0,1,2\}$ mod every $p_i$, decompose $a \bmod p_i$ as $0+0$, $1+0$, or $1+1$ for each $i$, and combine the first and the second components by the Chinese remainder theorem into $s, s' \in S^*_m$ with $s + s' = a$.

    For each $i \in [t]$, define $\varphi^{(i)}_0, \varphi^{(i)}_1, \varphi^{(i)}_2: \ZZ_{p_i} \to \ZZ$ by:
    \begin{align*}
        \varphi^{(i)}_0(z) &= -\one[z = -2] \\
        \varphi^{(i)}_1(z) &= \one[z = 0] + \one[z = -1] \\
        \varphi^{(i)}_2(z) &= -\one[z = 1].
    \end{align*}
    The relevant facts about $\varphi^{(i)}$ are the following:
    \begin{itemize}
        \item We have:
        \begin{equation}\label{eq:h0condition}
            \sum_{c = 0}^2 \varphi_c^{(i)}(0) = 1.
        \end{equation}
        Note that this relies on $p_i$ being odd.
        \item For $a \in \{1, 2\}$, we have:
        \begin{equation}\label{eq:weirdcancellation}
        \sum_{c = 0}^2 \varphi_c^{(i)}(z+ac) = -\one[z = -2] + \one[z+a = 0] + \one[z+a = -1] - \one[z+2a = 1] = 0.
        \end{equation}
        Indeed, when $a = 1$, the above holds because the RHS is $-\one[z = -2] + \one[z = -1] + \one[z = -2] - \one[z = -1] = 0$, and when $a = 2$ it holds because the RHS is $-\one[z = -2] + \one[z = -2] + \one[z = -3] - \one[z = -3] = 0$.
    \end{itemize}
    Both facts hold as identities of integer-valued functions on $\ZZ_{p_i}$.
    With this in mind, we define our $A$-decoder as follows: we index the terms by elements $\beta = (\beta_1, \ldots, \beta_t) \in \{0, 1, 2\}^t$, and define $b_{\beta} = \CRT(\beta_1, \ldots, \beta_t)$ and $\varphi_\beta(z) := \prod_{i \in [t]} \varphi_{\beta_i}^{(i)}(z) \bmod{m}$. Note that this is well-defined as each factor is periodic modulo a divisor of $m$. Now, for any $z \in \ZZ_m$ and $a \in \ZZ_m$, we have, modulo $m$,
    \begin{align*}
        \sum_{\beta \in \{0, 1, 2\}^t} \varphi_\beta(z+ab_\beta) &\equiv \sum_{\beta \in \{0, 1, 2\}^t} \prod_{i \in [t]} \varphi_{\beta_i}^{(i)}(z+ab_\beta) && \\
        &\equiv \sum_{\beta \in \{0, 1, 2\}^t} \prod_{i \in [t]} \varphi_{\beta_i}^{(i)}(z+a\beta_i) && \text{($b_\beta \equiv \beta_i \bmod p_i$)} \\
        &\equiv \prod_{i \in [t]} \left(\sum_{c \in \{0, 1, 2\}} \varphi_{c}^{(i)}(z+ac)\right). && \text{(distributivity)}
    \end{align*}
    Now, if $a \in A \backslash \{0\}$, there exists an $i \in [t]$ such that $a \equiv 1 \text{ or }2\pmod{p_i}$ (if $a \equiv 0$ modulo every prime factor then $a = 0$), and by~\eqref{eq:weirdcancellation} that factor in the above product is zero for every $z$, over the integers and hence modulo $m$.
    On the other hand, if $z = a = 0$, the above simplifies to $\prod_{i \in [t]} \left(\sum_{c \in \{0, 1, 2\}} \varphi_{c}^{(i)} (0)\right) = \prod_{i \in [t]} 1 = 1$ by~\eqref{eq:h0condition}. This verifies both conditions of~\cref{def:adecoderintro}, so this is an $A$-decoder of sparsity $\leq 3^t$.

    Finally, uniformity. Identify the index set $[3^t]$ with $\{0,1,2\}^t$ by base-$3$ counting, in space $O(t)$. Given $\beta$ and $z \in \ZZ_m$, the shift is computed as $b_\beta = \sum_{i \in [t]} \beta_i \cdot \tfrac{m}{p_i} \cdot \left(\left(\tfrac{m}{p_i}\right)^{-1} \bmod p_i\right) \bmod m$, where each inverse is obtained by the extended Euclidean algorithm; the value $\varphi_\beta(z)$ is computed by reducing $z$ modulo each $p_i$, evaluating $\varphi^{(i)}_{\beta_i}$ by comparing the residue against $0$, $\pm 1$, and $-2$, and multiplying the $t$ factors modulo $m$. Every intermediate quantity is at most $m^2$, so the space used is $O(\log m)$, and since $t \leq \log_3 m$, the total time is $\poly(\log m)$.
\end{proof}

\itpara{Construction 2: A full-residue $A$-decoder for an odd prime $m$.}

\begin{lemma}\label{s-decoding-CM}
    Let $m$ be an odd prime, given as input, and let $A = \ZZ_m$. Then, there exists a space-$O(\log m)$ uniform $A$-decoder over $\ZZ_m$ with sparsity $\leq m$.
\end{lemma}
\begin{proof}
    Index the terms by $j \in \ZZ_m$, and define $b_j = j$ and $\varphi_j: \ZZ_m \to \ZZ_m$ by:
    \[
    \varphi_j(z) = \one[z \neq 0] + \one[j = 0].
    \]
    Now for any $a \in \ZZ_m$ and any $z \in \ZZ_m$, we have:
    \begin{align*}
        \sum_{j \in \ZZ_m} \varphi_j(z+aj) &= 1 + \sum_{j \in \ZZ_m} \one[z+aj \neq 0] \\
        &= 1 + \left|\{j \in \ZZ_m: z+aj \neq 0\}\right|.
    \end{align*}
    When $a \not\equiv 0 \pmod{m}$, the map $j \mapsto z+aj$ is a bijection on $\ZZ_m$ because $m$ is prime, so the count is $m - 1$ and the sum is $m \equiv 0 \pmod{m}$.
    When $z = a = 0$, the sum is equal to $1$. It follows that this is an $A$-decoder as desired. Uniformity is immediate: given $j$ and $z$, the shift $b_j = j$ and the value $\varphi_j(z)$ are computed by two comparisons in space $O(\log m)$ (and time $\poly(\log m)$).
\end{proof}
\section{Main Result: Catalytic Tree Evaluation from Matching Vectors}
\label{sec:tep}

Throughout this section we work with the following objects, each presented by its parameters and uniform in the sense of~\cref{rmk:uniformaccess,def:uniformdecoder}:
\begin{itemize}
    \item a modulus $m > 1$ and a score set $S \subseteq \ZZ_m$ with $0 \in S$ and $S \cap (-S) = \{0\}$, i.e., $S$ is $2$-noncancelling (\cref{def:sumset});
    \item an $(S+S)$-decoder $(\varphi_1, \ldots, \varphi_k; b_1, \ldots, b_k)$ over $\ZZ_m$ (\cref{def:sumset,def:adecoderintro});
    \item an $S$-matching-vector family $\{\uv_x, \vecv_x \in \Z_m^d : x \in \zo^\ell\}$ of size $2^\ell$ (\cref{def:smvfoverview}); and
    \item an auxiliary \emph{output family} $\{\wv_y \in \Z_m^d : y \in \zo^\ell\}$, needed only by~\cref{thm:TEPonelevel}; the proof of~\cref{thm:algo} provides its own output families, so instantiations need not provide one.
\end{itemize}
We write $s$ for a common bound on the space of the access algorithms above, and assume $s \geq \log(dm)$. The noncancellation hypothesis $S \cap (-S) = \{0\}$ ensures that a sum of two scores vanishes only when both do, which is why the decoder is taken for the sumset $S+S$; recall from~\cref{lem:compositedecoder} that the composite construction provides exactly such a decoder for $S = S^*_m$. For the canonical set this hypothesis is automatic: if $s'$ and $-s'$ both lie in $S^*_m$ and $m$ is odd, then $s'$ is $0$ or $1$ and also $0$ or $-1$ modulo each prime factor, hence $s' \equiv 0$ modulo each prime factor, i.e., $s' = 0$.

Our main theorem is stated at the level of ring-valued registers; two corollaries then translate it into the catalytic model and into ordinary space.

\begin{theorem}[Main theorem, register form]\label{thm:algo}\label{thm:tep}
    In the setting above, there is an algorithm that is given, on its input tape, a \tephl instance (the truth table of a function $f_u: \zo^\ell \times \zo^\ell \to \zo^\ell$ for every $u\in \zo^{<h}$, and a value $v_u\in \zo^{\ell}$ for every $u\in \zo^h$), and, as global state, three registers $\boldtau_0, \boldtau_1, \boldtau \in \Z_m^d$ with arbitrary initial contents. The algorithm writes $v_{\emptyset}$ to its output tape and restores all three registers to their initial contents. It uses $O(\ell + h(\log k + \log m) + s)$ bits of workspace and runs in time $(O(k))^h \cdot \poly(2^{\ell}, k, d, 2^s)$.
\end{theorem}

\begin{corollary}[Catalytic form]\label{cor:catalytic}
    In the setting above, we can solve \tephl in $O(\ell + h(\log k + \log m) + s)$ free space, $O(d\log(dm))$ catalytic space, and $(O(k))^h \cdot \poly(2^{\ell}, k, d, 2^s)$ time.
\end{corollary}
\begin{proof}
    Interpret the catalytic tape as the three registers via~\cref{rmk:bitrep}: this costs a catalytic tape of length $O(d\log(dm))$, an additional $O(\log(dm)) = O(s)$ bits of free space for the stored shift, and an additive $\poly(dm)$ term in the runtime. Run the algorithm of~\cref{thm:algo} on the represented registers; it restores them modulo $m$, which by~\cref{rmk:bitrep} restores the tape exactly.
\end{proof}

\begin{corollary}[Ordinary-space form]\label{cor:ordinary}
    In the setting above, we can solve \tephl by an ordinary (non-catalytic) algorithm using $O(d\log m + \ell + h(\log k + \log m) + s)$ space and $(O(k))^h \cdot \poly(2^{\ell}, k, d, 2^s)$ time.
\end{corollary}
\begin{proof}
    Allocate the three registers in ordinary workspace, as $3d$ blocks of $\lceil \log m \rceil$ bits initialized to zero; every intermediate state is then a genuine element of $\Z_m^{3d}$, so no representation shift is needed, and~\cref{thm:algo} applies verbatim. In this form the final restoration pass of the algorithm may be skipped, though running it does not affect correctness; the restorations \emph{during} the recursion remain essential to correctness.
\end{proof}

The instantiations of this framework---the two Grolmusz-based settings behind the introduction's headline results, and a squared-distance setting recovering the Cook--Mertz--Goldreich algorithm~\cite{CM23,Gol25}---are collected in~\cref{sec:instantiations}; the extension of the framework to higher-arity trees, and the blocked bounds it yields, are in~\cref{sec:blocking}.

\subsection{Recursive Step: One Level of Tree Evaluation}
\label{sec:tep:step}

The main technical workhorse for our results is the following theorem. It is the formal counterpart of the informal one-level gadget $\olg$ of~\cref{def:olginformal}: the two families $y \mapsto \uv_y$ and $y \mapsto \vecv_y$ play the role of $\Enc$ on the input side, selected per query by the $\ctrl$ bit; the output family $\{\wv_y\}$ plays it on the output side; and the interface spells out what the informal definition elides---positional register addressing, and exactly how much state is retained across a query.
\begin{theorem}[One level]\label{thm:TEPonelevel}
    In the setting of this section, consider an algorithm given the following resources:
    \begin{itemize}
        \item global state comprising three registers $\boldtau_0, \boldtau_1, \boldtau \in \Z_m^d$, held in this fixed order; the oracle below addresses the registers by their position, regardless of how the algorithm may have modified or permuted their contents;
        \item the truth table of a function $f: \zo^\ell \times \zo^\ell \to \zo^\ell$ on the input tape; and
        \item access to an oracle $\calO$ that is parameterized by two labels $x_0, x_1 \in \zo^\ell$, which are fixed throughout the algorithm's execution and are not otherwise accessible to the algorithm. On query $(\alpha, \ctrl, \sigma) \in \ZZ_m \times \zo \times \zo$, given in local space, the oracle updates the register in position $\sigma$ as
        \[
        \boldtau_\sigma \gets \begin{cases} \boldtau_\sigma + \alpha\, \uv_{x_\sigma} & \text{if } \ctrl = 0,\\ \boldtau_\sigma + \alpha\, \vecv_{x_\sigma} & \text{if } \ctrl = 1,\end{cases}
        \]
        and leaves the other two registers unchanged.
    \end{itemize}
    Then there is an algorithm that, given $\alpha^* \in \ZZ_m$ in local space, performs the net update
    \[
    \boldtau \gets \boldtau + \alpha^* \cdot \wv_{f(x_0, x_1)}
    \]
    and leaves $\boldtau_0$ and $\boldtau_1$ unchanged, for \emph{every} initial state of the three registers; moreover, every update it makes to $\boldtau_0$ or $\boldtau_1$ is later undone by the exactly inverse update. It makes $4k + 4$ oracle queries. Between queries it uses $O(\ell + s + \log k)$ bits of local space, of which only $O(\log k + \log m)$ bits are retained across any single query; and it runs in time $\poly(2^{\ell}, k, d, 2^s)$.
\end{theorem}

\noindent
Before we prove the theorem, we begin with a preliminary lemma; the underlying compute-and-compare use of inner products is standard in catalytic computing~\cite{BCKLS14,CM23}.

\begin{lemma}\label{lemma:innerproduct}
    In the setting of~\cref{thm:TEPonelevel}, there is a procedure that computes the scalar
    \[
    g := \langle \boldtau_0, \vecv_{x_0}\rangle + \langle \boldtau_1, \vecv_{x_1} \rangle \in \ZZ_m,
    \]
    where $\boldtau_0, \boldtau_1$ denote the register contents at the start of the procedure, and stores it in local space. It uses four oracle queries and $O(\log d + \log m)$ local space, of which $O(\log m)$ bits are retained across queries; it leaves all three registers unchanged, undoing each of its updates exactly.
\end{lemma}
\begin{proof}
    The procedure is the following; recall that the oracle addresses registers by position.
    \begin{enumerate}
        \item Compute $\mathsf{tmp}_1 = \langle \boldtau_0, \boldtau_1 \rangle$ and write it into local space.
        \item Swap the contents of the first two registers, coordinate by coordinate. By position, the state is now $(\boldtau_1, \boldtau_0, \boldtau)$.
        \item Query $\calO(1, 1, 0)$: the oracle adds $\vecv_{x_0}$ to the register in position $0$. The state is now $(\boldtau_1 + \vecv_{x_0}, \boldtau_0, \boldtau)$.
        \item Compute $\mathsf{tmp}_2 = \langle \boldtau_1 + \vecv_{x_0}, \boldtau_0 \rangle$, the inner product of the first two registers, and write it into local space.
        \item Query $\calO(-1, 1, 0)$, restoring the state $(\boldtau_1, \boldtau_0, \boldtau)$.
        \item Query $\calO(1, 1, 1)$: the state is now $(\boldtau_1, \boldtau_0 + \vecv_{x_1}, \boldtau)$.
        \item Compute $\mathsf{tmp}_3 = \langle \boldtau_1, \boldtau_0 + \vecv_{x_1} \rangle$ and write it into local space.
        \item Query $\calO(-1, 1, 1)$, and swap the contents of the first two registers back.
    \end{enumerate}
    We have $\mathsf{tmp}_2 - \mathsf{tmp}_1 = \langle \vecv_{x_0}, \boldtau_0 \rangle$ and $\mathsf{tmp}_3 - \mathsf{tmp}_1 = \langle \boldtau_1, \vecv_{x_1} \rangle$, so the procedure outputs $g = (\mathsf{tmp}_2 - \mathsf{tmp}_1) + (\mathsf{tmp}_3 - \mathsf{tmp}_1)$. Each swap and each inner product of two registers is performed with a coordinate counter ($O(\log d)$ bits) and, for the inner products, an accumulator ($O(\log m)$ bits); across the oracle queries the procedure retains only the $\mathsf{tmp}$ values and its position in the list above, $O(\log m)$ bits in all. Every update made---the two swaps and the four oracle queries---is undone exactly.
\end{proof}

\subsection{Proof of~\ts{\cref{thm:TEPonelevel}}{One Level}}
\label{sec:tep:alg}

We give the algorithm, then its correctness and its efficiency analysis. Except where stated, all arithmetic is carried out modulo $m$.

\paragraph{Algorithm.} The input is a scalar $\alpha^* \in \ZZ_m$ in local space; the truth table of $f$ is on the input tape.
\begin{enumerate}
    \item Use~\cref{lemma:innerproduct} to compute and store $g = \langle \boldtau_0, \vecv_{x_0} \rangle + \langle \boldtau_1, \vecv_{x_1} \rangle$.
    \item Repeat for each term index $j = 1, 2, \ldots, k$:
    \begin{enumerate}
        \item Obtain the shift $b_j$ from the decoder's access algorithm, and make the two queries $\calO(b_j, 0, 0)$ and $\calO(b_j, 0, 1)$, i.e., the updates
        \begin{align*}
            \boldtau_0 &\gets \boldtau_0 + b_j \cdot \uv_{x_0}, \\
            \boldtau_1 &\gets \boldtau_1 + b_j \cdot \uv_{x_1}.
        \end{align*}
        \item Repeat for all candidate labels $y_0, y_1 \in \zo^\ell$:
        \begin{enumerate}
            \item Compute
            $$c := \varphi_j\left(\langle \boldtau_0, \vecv_{y_0} \rangle+ \langle \boldtau_1, \vecv_{y_1} \rangle -g\right) \in \ZZ_m.$$
            This is a direct computation involving no oracle query: the candidates $y_0, y_1$ are explicit, so the coordinates of $\vecv_{y_0}$ and $\vecv_{y_1}$ are generated by the family's access algorithm, the inner products with the (readable) register contents are accumulated coordinate by coordinate, and $\varphi_j$ is evaluated by the decoder's access algorithm.
            \item Read $f(y_0, y_1)$ from the truth table and apply, again directly and coordinate by coordinate, the update
            $$\boldtau \gets \boldtau + \alpha^* \cdot c \cdot \wv_{f(y_0, y_1)}.$$
        \end{enumerate}
        \item Make the two queries $\calO(-b_j, 0, 0)$ and $\calO(-b_j, 0, 1)$, restoring
        \begin{align*}
            \boldtau_0 &\gets \boldtau_0 - b_j \cdot \uv_{x_0}, \\
            \boldtau_1 &\gets \boldtau_1 - b_j \cdot \uv_{x_1}.
        \end{align*}
    \end{enumerate}
\end{enumerate}

\begin{lemma}[One-level correctness]\label{lem:onelevelcorrect}
    For every initial state of the three registers and every $\alpha^* \in \ZZ_m$, the algorithm above performs the net update $\boldtau \gets \boldtau + \alpha^* \cdot \wv_{f(x_0, x_1)}$ and leaves $\boldtau_0, \boldtau_1$ unchanged, undoing each of its updates to them exactly.
\end{lemma}
\begin{proof}
    The updates to $\boldtau_0$ and $\boldtau_1$ come in exactly inverse pairs: steps (a) and (c) of iteration $j$ add and subtract $b_j \uv_{x_\sigma}$, and the updates inside~\cref{lemma:innerproduct} are undone there. So these registers end unchanged.

    It remains to compute the net change to $\boldtau$. During iteration $j$ of step 2, the first two registers hold $\boldtau_0 + b_j \uv_{x_0}$ and $\boldtau_1 + b_j \uv_{x_1}$, where $\boldtau_0, \boldtau_1$ denote the initial contents, so the net change is
    \begin{align*}
        &\sum_{j \in [k]} \sum_{y_0, y_1 \in \zo^\ell} \alpha^* \cdot \varphi_j\left(\langle \boldtau_0 + b_j \uv_{x_0}, \vecv_{y_0} \rangle + \langle \boldtau_1 + b_j \uv_{x_1}, \vecv_{y_1} \rangle - g\right) \cdot \wv_{f(y_0, y_1)} \\
        = \ &\alpha^* \cdot \sum_{y_0, y_1 \in \zo^\ell} \wv_{f(y_0, y_1)} \cdot \left(\sum_{j \in [k]} \varphi_j\left(\left(\langle \boldtau_0, \vecv_{y_0} \rangle + \langle \boldtau_1, \vecv_{y_1} \rangle - g\right) + b_j\left( \langle \uv_{x_0}, \vecv_{y_0} \rangle + \langle \uv_{x_1}, \vecv_{y_1} \rangle\right)\right)\right),
    \end{align*}
    using linearity of the inner product.

    We evaluate the inner sums. By the matching-vector property (\cref{def:smvfoverview}), each score $\langle \uv_{x_\sigma}, \vecv_{y_\sigma}\rangle$ lies in $S$ (here $0 \in S$ covers the diagonal pairs) and vanishes exactly when $y_\sigma = x_\sigma$. Hence $\langle \uv_{x_0}, \vecv_{y_0} \rangle + \langle \uv_{x_1}, \vecv_{y_1} \rangle \in S + S$, and since $S \cap (-S) = \{0\}$, this sum is $0$ exactly when $(y_0, y_1) = (x_0, x_1)$. For $(y_0,y_1) \neq (x_0,x_1)$, then, it lies in $(S+S)\setminus\{0\}$, and the inner sum vanishes by the first condition of~\cref{def:adecoderintro}, which holds for every value of its other argument. For $(y_0,y_1) = (x_0,x_1)$, the quantity $\langle \boldtau_0, \vecv_{x_0} \rangle + \langle \boldtau_1, \vecv_{x_1} \rangle - g$ is $0$ by the definition of $g$, so the inner sum equals $\sum_{j} \varphi_j(0) = 1$ by the second condition. The net change to $\boldtau$ is therefore $\alpha^* \cdot \wv_{f(x_0,x_1)}$, as claimed.
\end{proof}

\paragraph{Efficiency analysis.} The oracle queries are the four in step 1 and the four per iteration of step 2, i.e., $4k + 4$ in total. The bits retained across a query are: within step 1, the input $\alpha^*$ and the $\mathsf{tmp}$ values and list position ($O(\log m)$, by~\cref{lemma:innerproduct}); within step 2, the term index $j$ ($O(\log k)$), the scalars $g, b_j, \alpha^*$ ($O(\log m)$), and the current position in the loop body ($O(1)$)---in all, $O(\log k + \log m)$. Between queries, step 2(b) additionally uses the candidates $y_0, y_1$ ($2\ell$ bits), a coordinate counter ($O(\log d)$), accumulators and the scalar $c$ ($O(\log m)$), a position on the input tape for the truth-table lookup ($O(\ell)$), and the workspace of the access algorithms ($O(s)$); since $s \geq \log(dm)$, the total is $O(\ell + s)$. For the runtime, step 2(b) dominates: it runs $k \cdot 2^{2\ell}$ times, and each run makes $O(d)$ coordinate steps, each involving access computations that run in time $2^{O(s)}$ (\cref{rmk:uniformaccess}); the total is $\poly(2^{\ell}, k, d, 2^s)$. \qed

\subsection{Proof of~\ts{\cref{thm:algo}}{Multiple Levels}}\label{sec:gadgettofull}

The recursion in this subsection is the one introduced by Cook and Mertz~\cite{CM23} (see also the exposition of Goldreich~\cite{Gol25}): evaluate the root by the one-level procedure, serving each of its oracle queries by recursing on the corresponding child. We prove it in full rather than citing it, because our resource bounds rest on details that differ from~\cite{CM23}: the oracle contract of~\cref{thm:TEPonelevel} (positional registers, two hidden labels), the accounting of input access and of the state retained across queries, and the exact restoration that~\cref{cor:catalytic} extracts via~\cref{rmk:bitrep}.

Throughout, nodes are indexed by $u \in \zo^{\leq h}$ with the root $\emptyset$ at depth $0$ and the leaves $u \in \zo^h$ at depth $h$, and $v_u$ denotes the value of the \tephl instance at node $u$. We assume without loss of generality that $d \geq E := \lceil \ell/\lfloor \log m \rfloor\rceil$, padding every vector of the matching-vector family with zero coordinates otherwise (this preserves all hypotheses---the access algorithms append the zero coordinates with a counter---and increases $d$ by at most $E$). For $y \in \zo^\ell$, let $\enc(y) \in \Z_m^d$ denote the vector whose first $E$ coordinates hold the $\lfloor \log m\rfloor$-bit blocks of $y$, each read as an element of $\ZZ_m$, and whose remaining coordinates are $0$.

The recursion invokes~\cref{thm:TEPonelevel} with three output families: $\{\uv_y\}$ and $\{\vecv_y\}$, whose access algorithms are those of the matching-vector family, and $\{\enc(y)\}$, whose coordinates are computed directly from $y$ in space $O(\ell + \log(dm)) = O(\ell + s)$.

\paragraph{Workspace allocation.} The algorithm maintains:
\begin{itemize}
    \item a node register holding the current node $u \in \zo^{\leq h}$ and its depth ($O(h)$ bits);
    \item for each depth $0 \leq i < h$, a slot of $O(\log k + \log m)$ bits holding the retained state of the depth-$i$ call on the call stack, saved while that call waits for a recursive call to return (the calls on the stack lie along a single root-to-node path, so there is at most one per depth): its term index, scalars, phase, output-family selector, and pending query parameters;
    \item a single shared block of $O(\ell + h + s)$ bits for the scratch space of the currently active call (the $h$ covers positions on the global input tape; see below); and
    \item a readout register $\reg$ of $E$ elements of $\ZZ_m$ ($O(\ell + \log m)$ bits).
\end{itemize}
The sharing in the third item is the point of~\cref{thm:TEPonelevel}'s guarantee that only $O(\log k + \log m)$ bits are retained across a query: a call uses its scratch space only between its own queries, and while a query is being served---which is when a child call runs---that call's scratch space is dead and the block can be reused by the child. The total is $O(\ell + h(\log k + \log m) + s)$ bits.

\paragraph{Input access.} \Cref{thm:TEPonelevel} speaks of an input tape holding a single truth table; the recursion gives the call at node $u$ the table $f := f_u$ by address translation. The base offset of $f_u$'s truth table within the \tephl input is computed from $u$ with an $O(h + \ell)$-bit counter, counted in the shared scratch block, and each access moves the input head across a tape of length $2^{O(h+\ell)}$; the runtime recurrence below accounts for these seeks. Leaf values $v_{u\sigma}$ are located in the same way.

\paragraph{Handling oracle queries.} Suppose the one-level call at an internal node $u$ makes the query $(\alpha, \ctrl, \sigma)$; the hidden labels of its oracle are the children's values $x_0 = v_{u0}$ and $x_1 = v_{u1}$.
If $|u| = h-1$, the child $u\sigma$ is a leaf: the algorithm reads $v_{u\sigma}$ directly from the input tape and applies the update $\boldtau_\sigma \gets \boldtau_\sigma + \alpha \cdot \uv_{v_{u\sigma}}$ (or with $\vecv_{v_{u\sigma}}$ if $\ctrl = 1$) coordinate by coordinate, generating the vector with the family's access algorithm, in scratch space $O(\ell + h + s)$.
Otherwise, the algorithm: saves the call's retained bits in the depth-$|u|$ slot; swaps the contents of the registers in positions $\sigma$ and $2$ (coordinate by coordinate, with an $O(\log d + \log m)$-bit counter); updates the node register to $u\sigma$; and starts a fresh one-level call at $u\sigma$ with input $\alpha^* := \alpha$ and output family $\{\uv_y\}$ if $\ctrl = 0$, $\{\vecv_y\}$ if $\ctrl = 1$. When that call returns, it swaps positions $\sigma$ and $2$ back, restores the node register to $u$, and resumes from the saved slot.

\paragraph{Correctness.} We claim, by induction on the height of the subtree below $u$: a one-level call at an internal node $u$, run with input $\alpha^*$ and any of the three output families $\{\wv_y\}$, changes the register in position $2$ by $+\alpha^* \cdot \wv_{v_u}$ and leaves the other two registers unchanged---for every state of the registers at the start of the call. Indeed, the oracle contract of~\cref{thm:TEPonelevel} with hidden labels $v_{u0}, v_{u1}$ is implemented faithfully: at $|u| = h-1$ directly from the input tape, and at $|u| < h-1$ by the child call, since after swapping positions $\sigma$ and $2$, the child (by the inductive hypothesis, with $v_u$ replaced by $v_{u\sigma}$) adds $\alpha \cdot \uv_{v_{u\sigma}}$ or $\alpha \cdot \vecv_{v_{u\sigma}}$ to position $2$---which at that moment holds the contents of $\boldtau_\sigma$---and leaves the other positions unchanged, so that after swapping back, the net effect is exactly the queried update, and note that $f_u(v_{u0}, v_{u1}) = v_u$. Given the oracle,~\cref{lem:onelevelcorrect} yields the claim at $u$.

\paragraph{Readout and restoration.} At the top level, the algorithm first copies the first $E$ coordinates of the register in position $2$ into $\reg$. It then makes the one-level call at the root with $\alpha^* = 1$ and output family $\{\enc(y)\}$; by the claim above, this adds $\enc(v_\emptyset)$ to position $2$. The algorithm reads the first $E$ coordinates of position $2$ and, in local space, subtracts the stored $\reg$ values modulo $m$, decodes $v_\emptyset$ from the differences, and writes it to the output tape; the registers are not modified by this read. Finally it makes the call at the root once more, with $\alpha^* = -1$ and the same output family, which adds $-\enc(v_\emptyset)$ to position $2$; all three registers now hold their initial contents, as~\cref{thm:algo} requires. (Note that the second run is correct regardless of the first: the claim above holds for every starting state of the registers.)

\paragraph{Runtime.} For $1 \leq \eta \leq h$, let $T(\eta)$ bound the time of a one-level call at a node $u$ with $\eta = h - |u|$. By~\cref{thm:TEPonelevel}, the call's own work is $2^{O(h)} \cdot \poly(2^\ell, k, d, 2^s)$, where the $2^{O(h)}$ factor accounts for the input-tape seeks of the address translation above, and each of its $4k+4$ queries costs $2^{O(h)} \cdot \poly(2^{\ell}, d, 2^s)$ for a leaf update ($\eta = 1$) or two register swaps plus $T(\eta - 1)$ otherwise. Hence
\[
T(\eta) \leq (4k+4) \cdot T(\eta-1) + 2^{O(h)} \cdot \poly(2^{\ell}, k, d, 2^s), \qquad T(1) = 2^{O(h)} \cdot \poly(2^{\ell}, k, d, 2^s),
\]
so $T(h) = (O(k))^h \cdot \poly(2^{\ell}, k, d, 2^s)$, the $2^{O(h)}$ factors being absorbed into $(O(k))^h$. The two root runs and the readout change this by at most a constant factor plus lower-order additive terms, giving the stated runtime. \qed

\subsection{Instantiations}\label{sec:instantiations}

We instantiate the framework three ways. The first two use Grolmusz's matching vectors with the composite-modulus decoder and give the catalytic results of the introduction; the third uses the squared-distance vectors with the full-residue decoder and recovers the guarantees of Cook and Mertz~\cite{CM23} within the same framework, with no catalytic tape. In each case we specify the matching-vector family, the decoder, the modulus, and the access-space bound $s$, and read off the resources from~\cref{cor:catalytic} or~\cref{cor:ordinary}.

\begin{corollary}\label{cor:eps}
    For any $\epsilon > 0$, we can solve \tephl in $O(\ell + h)$ free space, $\exp(O(\ell^\epsilon))$ catalytic space, and $\poly(2^{\ell+h})$ time.
\end{corollary}
\begin{proof}
    Let $t = \max(2, \lceil 3/\epsilon \rceil)$, let $p_1, \ldots, p_t$ be the first $t$ odd primes, and let $m = p_1 \cdots p_t$ and $S = S^*_m$, so that $m$ and the decoder sparsity $k \leq 3^t$ (\cref{lem:compositedecoder}) are constants. Instantiate~\cref{cor:catalytic} with the matching-vector family of~\cref{item:grolmusz} of~\cref{cor:mvspecific}, which has size exactly $2^\ell$ and dimension $d = \exp\left(O_m\left(\ell^{1/t}(\log(\ell+2))^{1-1/t}\right)\right) \leq \exp(O(\ell^{\epsilon}))$. All access algorithms run in space $O(\ell + \log d + \log m) = O(\ell)$, so we may take $s = O(\ell)$. The resources are then $O(\ell + h)$ free space, $O(d \log (dm)) \leq \exp(O(\ell^\epsilon))$ catalytic space, and $2^{O(h)} \cdot \poly(2^\ell) = \poly(2^{\ell+h})$ time.
\end{proof}

\begin{corollary}\label{cor:sqrt}
    Suppose $\ell$ is larger than a sufficiently large absolute constant. Then we can solve \tephl in $O(\ell+h\sqrt{\log \ell}\log \log \ell)$ free space, $\exp(\exp(O(\sqrt{\log \ell})))$ catalytic space, and $\poly(2^{\ell+h\sqrt{\log \ell}})$ time.
\end{corollary}
\begin{proof}
    Let $t = \lceil \sqrt{\log \ell}\,\rceil$ and let $p_1, \ldots, p_t$ be the first $t$ odd primes. Then $\log m = O(t\log t)$, which is $O(\sqrt{\log \ell}\,\log\log\ell)$, and the decoder of~\cref{lem:compositedecoder} has sparsity $k \leq 3^t = 2^{O(\sqrt{\log \ell})}$. Instantiate~\cref{cor:catalytic} with the family of~\cref{item:manyprimes} of~\cref{cor:mvspecific}, which has $d = \exp(\exp(O(\sqrt{\log\ell})))$; note $\log d = 2^{O(\sqrt{\log \ell})} = O(\ell)$ for $\ell$ large, so all access algorithms run in space $O(\ell)$ and we may take $s = O(\ell)$. The free space is then $O(\ell + h(\log k + \log m) + s) = O(\ell + h\sqrt{\log\ell}\,\log\log\ell)$. The catalytic space is $O(d\log(dm)) = \exp(\exp(O(\sqrt{\log\ell})))$. The time is $2^{O(h\sqrt{\log\ell})} \cdot \poly(2^\ell, d) = \poly(2^{\ell + h\sqrt{\log\ell}})$. (For $\ell$ below the constant, \cref{cor:eps} with $\epsilon = 1$ already gives $O(h)$ free space and $O(1)$ catalytic space.)
\end{proof}

\begin{corollary}[Cook--Mertz--Goldreich~\cite{CM23,Gol25}]\label{cor:packedCM}
    We can solve \tephl by an ordinary (non-catalytic) algorithm using $O(\ell + h\log(\ell+2))$ space and $2^{O(\ell + h\log(\ell+2))}$ time. In particular, \tep is solvable in $O(\log n \log\log n)$ space and $n^{O(\log\log n)}$ time, the bounds of~\cite{CM23}.
\end{corollary}
\begin{proof}
    Set the digit range $M := \max(2, \ell)$ and the number of digits $w := \lceil \ell/\lfloor \log M\rfloor\rceil$, so that $M^w \geq 2^{\ell}$, and let $D := wM^2$, which bounds the squared-distance scores. Using~\cref{lem:primeselect}, generate the least prime $m \in (2D, 4D)$; since $m > 2D \geq 8$, the prime $m$ is odd and exceeds $wM^2$, as \cref{s-matching-CM,s-decoding-CM} require. All of $M, w, D, m$ are computed from $\ell$ once, at the start, in space $O(\log D)$ and time $\poly(D)$. We instantiate the setting of this section as follows.
    \begin{itemize}
        \item Scores: $S := \{0, 1, \ldots, D\}$, which contains $0$; since $m > 2D$, we have $S \cap (-S) = \{0\}$ (\cref{rmk:squarednoncancel}).
        \item Family: the squared-distance family of~\cref{s-matching-CM} over $\Z_m^{w+2}$, of size $M^w \geq 2^{\ell}$, restricted to its first $2^{\ell}$ indices; a restriction changes neither the matching property nor the uniformity. Thus $d = w + 2$.
        \item Decoder: the full-residue decoder of~\cref{s-decoding-CM}, of sparsity $k \leq m$. Its first defining condition holds for every nonzero residue, so in particular for every element of $(S+S)\setminus\{0\}$, and since the second condition does not involve the set $A$, it is therefore an $(S+S)$-decoder. (The score set certified by~\cref{s-matching-CM} is the whole interval $\{0, \ldots, D\}$, so the sumset is all of $\{0, \ldots, 2D\}$; this is why we use the full-residue decoder here.)
        \item Access space: the family's algorithm runs in space $O(\log(M^w) + \log m) = O(\ell + \log(\ell+2))$ and the decoder's in $O(\log m)$, so we may take $s = O(\ell + \log(\ell+2)) = O(\ell)$, and $s \geq \log(dm)$ holds.
    \end{itemize}
    Now apply~\cref{cor:ordinary}. We have $d\log m = O\!\left(\left(\tfrac{\ell}{\log(\ell+2)} + 1\right)\log(\ell+2)\right) = O(\ell + \log(\ell+2))$, since $\lfloor \log M \rfloor = \Theta(\log(\ell+2))$, and $\log k + \log m = O(\log(\ell+2))$, so the space is
    \begin{align*}
    O\!\left(d\log m + \ell + h(\log k + \log m) + s\right) &= O(\ell + h\log(\ell+2)). && \text{($h \geq 1$ absorbs the lone $\log(\ell+2)$ terms)}
    \end{align*}
    The time is
    \[
    (O(k))^h \cdot \poly(2^{\ell}, k, d, 2^s) = 2^{O(h\log(\ell+2))} \cdot 2^{O(\ell)} = 2^{O(\ell + h\log(\ell+2))}.
    \]
    For the final sentence, recall that $h, \ell \leq O(\log n)$, so $\log(\ell+2) = O(\log\log n)$: the space is $O(\log n \log\log n)$ and the time is $2^{O(\log n \log\log n)} = n^{O(\log\log n)}$.
\end{proof}

\paragraph{Better tree evaluation from better ingredients.}
Improvements to the framework's two combinatorial ingredients enter the bounds only through the parameters $d$, $m$, $k$, and $s$ of~\cref{cor:catalytic,cor:ordinary}. The most consequential case is the composite regime, where the dimension $d$ is the bottleneck.

\begin{corollary}[Tree evaluation from optimal matching vectors]\label{cor:betterMV}
    Fix $m = p_1 \cdots p_t$, a product of $t \geq 2$ distinct odd primes, and suppose that there is a space-$O(\ell)$-uniform sequence of $S^*_m$-matching-vector families whose $\ell$th member has size $2^\ell$ and lies in $\Z_m^d$ with $d = O(\ell\log(\ell+2))$. Then \tep is solvable in $O(\log n \log\log n)$ ordinary space and polynomial time.
\end{corollary}
\begin{proof}
    Apply~\cref{cor:ordinary} with this family and the decoder of~\cref{lem:compositedecoder}, whose sparsity $k \leq 3^t$ is constant, and take $s = O(\ell)$. The space is $O(d\log m + \ell + h(\log k + \log m) + s) = O(\ell\log(\ell+2) + h)$ and the time is $(O(k))^h \cdot \poly(2^\ell, k, d, 2^s) = 2^{O(h + \ell)}$. Since $h, \ell \leq O(\log n)$, these are $O(\log n\log\log n)$ and $\poly(n)$.
\end{proof}

\noindent This would match the space bound of~\cite{CM23} with a polynomial running time, and with no catalytic tape at all. The corollary's dimension hypothesis matches, up to a constant factor, the following lower bound. Bhowmick, Dvir, and Lovett~\cite{DBLP:journals/siamcomp/0001DL14} showed that, under the polynomial Freiman--Ruzsa conjecture~\cite{ruzsa1999analog} over $\Z_m^d$ with $m$ constant, every matching-vector family in $\Z_m^d$ has size $2^{O(d/\log d)}$, so a family of size $2^\ell$ requires $d \geq \Omega(\ell\log\ell)$; the conjecture has since been proven by Gowers, Green, Manners, and Tao---over $\F_2^d$ in~\cite{gowers2025conjecture}, and in the bounded-torsion generality needed here in~\cite{gowers2024conjecture}---so this bound is unconditional. Conversely, to the best of our knowledge, no known lower bound~\cite{DBLP:journals/cc/TardosB98,DBLP:journals/siamcomp/0001DL14,DBLP:conf/innovations/AggarwalDLOS25} rules out families meeting the hypothesis, whereas the best known constructions (\cref{cor:mvspecific}) have $d = \exp\big(O(\ell^{1/t}(\log(\ell+2))^{1-1/t})\big)$.

\subsection{Higher Arity and Blocking}\label{sec:blocking}

The framework extends from binary to $r$-ary trees with only local changes, and the extension buys an improvement in space: grouping levels of a binary tree into one higher-arity level rebalances the one-time cost of the registers against the per-level cost of the bookkeeping, an idea of Stoeckl~\cite{Stoeckl} and Goldreich~\cite{Gol24TE}. We state the extension, then relate it to the blocked bounds. We use the $r$-fold sumsets $rS$ and $r$-noncancellation of~\cref{def:sumset}.

\begin{theorem}[Higher arity, register form]\label{thm:rary}
    Let $r \geq 2$, and modify the setting of this section as follows: the score set $S$ is $r$-noncancelling (\cref{def:sumset}; for $r = 2$ this is $S \cap (-S) = \{0\}$); the decoder is an $rS$-decoder over $\ZZ_m$; and the global state comprises $r+1$ registers $\boldtau_0, \ldots, \boldtau_{r-1}, \boldtau \in \Z_m^d$. Then there is an algorithm that is given, on its input tape, a \tephlr instance, writes the root value to its output tape, and restores all $r+1$ registers to their initial contents. It uses $O(r\ell + h(\log k + \log m + \log r) + s)$ bits of workspace and runs in time $2^{O(h\log r)} \cdot (O(k))^h \cdot \poly(2^{r\ell}, k, d, 2^s)$.
\end{theorem}

\begin{proof}[Proof sketch]
    The binary proof carries over with the following changes; $r = 2$ recovers \cref{thm:TEPonelevel,thm:algo} verbatim.

    \emph{One level.} The oracle of~\cref{thm:TEPonelevel} is now parameterized by $r$ hidden labels $x_0, \ldots, x_{r-1}$, and a query $(\alpha, \ctrl, \sigma)$ with $\sigma \in \{0, \ldots, r-1\}$ updates the register in position $\sigma$. The algorithm first computes $g := \sum_{i} \langle \boldtau_i, \vecv_{x_i}\rangle$ by the swap trick of~\cref{lemma:innerproduct}, pairing the register in position $i$ with the one in position $i+1 \bmod r$ (two queries per child, $2r$ in all). Then, for each term index $j \in [k]$---a single index, shared by all $r$ children---it adds $b_j \uv_{x_i}$ to every child register ($r$ queries), loops over all candidate tuples $(y_0, \ldots, y_{r-1}) \in (\zo^\ell)^r$ performing the direct update
    \[
    \boldtau \gets \boldtau + \alpha^* \cdot \varphi_j\Big(\sum\nolimits_i \langle \boldtau_i, \vecv_{y_i}\rangle - g\Big)\cdot \wv_{f(y_0, \ldots, y_{r-1})},
    \]
    and undoes the $r$ additions, for $2r(k+1)$ queries in all. In the telescoping computation of~\cref{lem:onelevelcorrect}, the score sum $\sum_i \langle \uv_{x_i}, \vecv_{y_i}\rangle$ lies in $rS$, and it is $0$ exactly when every $y_i = x_i$, by the matching property and $r$-noncancellation; off the diagonal the inner sum vanishes by the first condition of the $rS$-decoder, and on the diagonal $g$ cancels the correction and the inner sum is $1$. The bits retained across a query are $O(\log k + \log m + \log r)$ (term index, scalars, phase, and the pending query's $\sigma$); the scratch space between queries is $O(r\ell + s)$ (the candidate tuple, counters, accumulators, and the truth-table position); the call's own time is $\poly(2^{r\ell}, k, d, 2^s)$.

    \emph{Recursion.} As in~\cref{sec:gadgettofull}: the node register costs $O(h\log r)$ bits, each depth gets a slot of $O(\log k + \log m + \log r)$ bits, the shared scratch block is $O(r\ell + h\log r + s)$ bits (the input tape now has length $2^{O(h\log r + r\ell)}$, with the same address translation), and a query $(\alpha, \ctrl, \sigma)$ at an internal node is served by swapping positions $\sigma$ and $r$ and recursing on child $\sigma$. The readout via $\enc$ and the two root runs are unchanged. The workspace totals as stated, and the recurrence $T(\eta) \leq 2r(k+1)\, T(\eta-1) + 2^{O(h\log r)} \cdot \poly(2^{r\ell}, k, d, 2^s)$ solves to the stated time: the coefficient contributes $(2r(k+1))^h = 2^{O(h\log r)}\cdot(O(k))^h$, and the input-tape seek factors contribute another $2^{O(h\log r)}$---which, unlike the binary case's $2^{O(h)}$, cannot be absorbed into the base of the power, so it stays explicit. Finally, the translations of~\cref{cor:catalytic,cor:ordinary} apply with $r+1$ registers in place of three (the argument of~\cref{rmk:bitrep} is unchanged with $(r+1)d$ blocks); in particular, allocating the registers in initialized ordinary workspace costs $O(rd\log m)$ bits of additional space, the form used below.
\end{proof}

\begin{theorem}[Stoeckl~\cite{Stoeckl}; {\cite[Theorem 3]{Gol24TE}}]\label{thm:blocked}
    For every $1 \leq b \leq h$, \tephl can be solved in space $O\big(h + 2^b \ell + \tfrac{h}{b}\log(\ell+2)\big)$.
\end{theorem}
\noindent The proof groups each run of $b$ consecutive levels into a single node of fan-in $2^b$, whose gate is evaluated on demand from the original binary gates, and applies the bound $O(r\ell + h\log(r\ell))$ for fan-in-$r$ trees~\cite[Theorem 7]{CM23} (see also~\cite[Theorem 2]{Gol24TE}) with $r = 2^b$. \Cref{thm:rary} recovers that fan-in-$r$ bound within our framework whenever $r \leq \poly(\ell)$: instantiate it as in~\cref{cor:packedCM}, but with the prime $m$ chosen in $(rD, 2rD)$, so that $S$ is $r$-noncancelling: a sum of $r$ scores cannot vanish modulo $m$ unless every score does. The registers then cost $O(r\ell)$ bits, and each level of the recursion retains $O(\log(r\ell))$ bits. Taking $b = \lfloor 0.9 \log\log\ell \rfloor$ in~\cref{thm:blocked} gives the space bound $O(\log n\log\log n/\log\log\log n)$~\cite{Stoeckl,Gol24TE}; since this block size has $2^b \leq \ell$, our framework recovers this bound as well.

\begin{remark}[Sum versus product]\label{rmk:sumproduct}
    Note that \Cref{thm:rary} decodes the \emph{sum} of the $r$ children's scores; this is what the $rS$-decoder and $r$-noncancellation are for. Applying an $S$-decoder to each child's masked score separately and multiplying the $r$ results, as the one-hot algorithm of~\cite{CM20} does for $r = 2$, would avoid both, and also isolate the diagonal term. But the product of $r$ sums expands into $k^r$ terms, indexed by tuples $(j_1, \ldots, j_r) \in [k]^r$ of \emph{independent} term indices, one per child; as written, the recursion retains $r\log k$ bits per level (the current tuple) and makes $\Theta(rk^r)$ oracle queries per call, against $O(\log k + \log m + \log r)$ bits and $2r(k+1)$ queries for the sumset route. For $r = 2$ the difference is a constant factor in the space bound and in the exponent of the time bound, but for blocking it matters: with $r = 2^b$, grouping $b$ levels divides the number of levels by $b$ while multiplying the per-level bookkeeping by $2^b$, erasing the savings of~\cref{thm:blocked}. The shared term index of the sumset route avoids this, so reducing the tree to $h/b$ levels does save space. (We compare against this natural implementation of the product route, but it is conceivable some other implementation could avoid these overheads.)
\end{remark}

\section{Application: New Time-Space-Catalytic Space Tradeoffs}
\label{sec:timespace}

We present two applications of our tree-evaluation algorithms: time-space tradeoffs in the catalytic model via the reduction of Williams~\cite{Wil25}, and analogous tradeoffs for circuit evaluation via the reduction of Shalunov~\cite{Sha25}.

\subsection{Composing with Space-Bounded Reductions}
\label{sec:timespace:compose}

The reductions in this section produce tree-evaluation instances of size exponential in the resources of interest, so no algorithm can afford to write one down. Instead, an instance $I$ is \emph{implicitly presented in space $s_P$}: a procedure $P$, given the index of a bit of $I$, outputs that bit using $O(s_P)$ bits of workspace and $2^{O(s_P)}$ time. Our algorithms run on such instances by the standard method of emulative composition~\cite[Lemma 5.2]{Gol08} (see also~\cite[Proposition 3.3]{PT25}): whenever the algorithm reads a bit of its input, run $P$ to compute that bit. This adds $O(s_P)$ bits of workspace and multiplies the running time by $2^{O(s_P)}$; the index of the bit read is held by the algorithm itself, whose space bounds already account for input addressing (\cref{sec:gadgettofull}). Since $P$ runs in ordinary workspace and never touches the catalytic tape, the catalytic guarantees of our algorithms are unaffected, and the root value they compute may be kept in an $O(\ell)$-bit buffer instead of being written to the output tape.

\subsection{The Reduction of Williams}
\label{sec:timespace:williams}

We now state the reduction of Williams~\cite{Wil25} from time-bounded computation to tree evaluation. Note that the simulation of~\cite{Wil25} does \emph{not} reduce a computation to a single tree-evaluation instance, since we cannot compute the pattern of tape head movements in advance. Instead it produces a sequence of TEP instances, where each candidate pattern $\sigma$ yields an (implicitly represented) instance that either computes the answer or declares failure at the root. Since $t$ need not be time-constructible, he also enumerates candidate time bounds $T$. His construction gives the following.

\begin{theorem}[{\cite[Section 3.2]{Wil25}}]\label{thm:Wil}
    Let $M$ be a multitape Turing machine deciding a language $L$ in time $t(n) \geq n$, and let $b(T)$ be computable from $T$ in space $O(\log T)$. There are constants $r = r(M) \geq 2$ and $c = c(M)$ such that the following holds on every input $x \in \zo^n$. For every $T \geq n$ with $\log T \leq b(T) \leq T$ and every $\sigma \in \zo^{\leq cT/b(T)}$, there is a $\mathsf{TreeEval}_{h,\ell,r}$ instance $I_{x,T,\sigma}$, with $h := \lceil cT/b(T)\rceil$ and $\ell := \lceil c\, b(T)\rceil$, implicitly presented in space $O(T/b(T) + b(T))$ given $(x, T, \sigma)$, whose root value is either a designated string $\bot$, signaling failure (in particular whenever the simulated computation has not halted within $T$ steps), or else contains a designated \emph{verdict bit}, such that:
    \begin{enumerate}
        \item if $T \geq t(n)$, then for at least one $\sigma$ the root value of $I_{x,T,\sigma}$ is not $\bot$; and
        \item whenever the root value of $I_{x,T,\sigma}$ is not $\bot$, its verdict bit equals $L(x)$.
    \end{enumerate}
\end{theorem}

Combining the reduction with our catalytic algorithms yields the tradeoffs. The first recovers the parameters announced in the introduction.

\begin{corollary}\label{cor:TS}
    Let $L$ be decided by a multitape Turing machine in time $t(n) \geq n$. For every $\epsilon > 0$,
    \[
    L \in \CT{2^{O(t^{\epsilon})}}{O(\sqrt{t}\,)}{2^{O(\sqrt{t}\,)}},
    \]
    which is \cref{cor:TSintro}.
\end{corollary}
\begin{proof}
    Apply~\cref{thm:Wil} with $b(T) := \lceil \sqrt{T}\,\rceil$, which satisfies $\log T \leq b(T) \leq T$ for $T \geq 16$; the finitely many inputs with $n < 16$ are decided trivially. The algorithm enumerates $T = n, n+1, \ldots$ and, for each $T$, the strings $\sigma \in \zo^{\leq cT/b(T)}$ in lexicographic order, holding $T$ and $\sigma$ in $O(\log T + T/b(T)) = O(\sqrt T)$ bits of free space. For each pair, the instance $I_{x,T,\sigma}$ has $h, \ell = O(\sqrt T)$ and constant fan-in $r$, and is implicitly presented in space $O(\sqrt T)$; composing its presentation with~\cref{lem:fanintwo} implicitly presents, still in space $O(\sqrt T)$, a binary $\mathsf{TreeEval}_{h', \ell'}$ instance with $h', \ell' = O(\sqrt T)$ whose root value is that of $I_{x,T,\sigma}$ padded with zeros. Run the algorithm of~\cref{cor:eps}, with its parameter set to $\min(1, 2\epsilon)$, on this instance by emulative composition, capturing the root value in a workspace buffer; its first $\ell$ bits are the root value of $I_{x,T,\sigma}$ and the rest are zeros, by~\cref{lem:fanintwo}. If those $\ell$ bits are $\bot$, move on to the next pair; otherwise output the verdict bit and halt. By property 1 of~\cref{thm:Wil}, the algorithm halts by the time $T$ reaches $t(n)$, and by property 2 its output is $L(x)$.

    Resources: every $T$ considered is at most $t = t(n)$. The free space is $O(\sqrt{t})$ for $T$, $\sigma$, the presentations, and the parameter computations, plus $O(\ell' + h') = O(\sqrt t)$ from~\cref{cor:eps}. The catalytic space is $\exp(O(\ell'^{\min(1, 2\epsilon)})) = 2^{O(t^{\epsilon})}$; the tape is restored at the end of each run, so a single tape serves every iteration and ends in its initial state. The time is at most $t \cdot 2^{O(\sqrt t)} = 2^{O(\sqrt t)}$ iterations, each costing $\poly(2^{\ell' + h'}) \cdot 2^{O(\sqrt t)} = 2^{O(\sqrt t)}$ by emulative composition.
\end{proof}

Optimizing the block length against the free space of~\cref{cor:sqrt} instead gives a sharper catalytic-space bound, at the cost of requiring more free space (albeit still less than that of~\cite{Wil25}).

\begin{corollary}\label{cor:TSsqrt}
    Let $L$ be decided by a multitape Turing machine in time $t(n) \geq n$. Then
    \[
    L \in \CT{\exp(\exp(O(\sqrt{\log t})))}{O\big(\sqrt{t}\,(\log t)^{1/4}\sqrt{\log\log t}\,\big)}{2^{O(\sqrt{t}\,(\log t)^{1/4}\sqrt{\log\log t}\,)}}.
    \]
\end{corollary}
\begin{proof}
    Repeat the proof of~\cref{cor:TS} with $b(T) := \lceil \sqrt{T}\,(\log T)^{1/4}(\log\log T)^{1/2}\rceil$ and~\cref{cor:sqrt} in place of~\cref{cor:eps}; the guard of~\cref{cor:sqrt} holds because $\ell' = \Theta(b(T))$ exceeds any absolute constant once $n$ is large enough, and smaller inputs are decided trivially. Fix a guess $T \leq t$ and write $b := b(T)$. Since $\log \ell' = \Theta(\log T)$ and $\log\log \ell' = \Theta(\log\log T)$, the free space is
    \begin{align*}
    &O\big(T/b + b\big) + O\big(\ell' + h'\sqrt{\log \ell'}\,\log\log \ell'\big) \\
    &\quad = O\big(b + (T/b)\sqrt{\log T}\,\log\log T\big) && \text{($\ell' = \Theta(b)$, $h' = O(T/b)$, $T/b \leq b$)} \\
    &\quad = O\big(\sqrt{T}\,(\log T)^{1/4}\sqrt{\log\log T}\,\big), && \text{(the two terms balance at the chosen $b$)}
    \end{align*}
    which is at most the displayed bound since $T \leq t$. The catalytic space is $\exp(\exp(O(\sqrt{\log \ell'}))) = \exp(\exp(O(\sqrt{\log t})))$. For each $T$ there are $2^{O(T/b)}$ iterations, each costing $\poly(2^{\ell' + h'\sqrt{\log \ell'}}) \cdot 2^{O(b)}$; since $h'\sqrt{\log \ell'} = O((T/b)\sqrt{\log T}) = O(b/\log\log T)$, each iteration is $2^{O(b)}$, and summing over the at most $t$ guesses gives the displayed time.
\end{proof}

\subsection{Circuit Evaluation}
\label{sec:timespace:circuit}

Shalunov~\cite{Sha25} subsequently gave a direct reduction from circuit evaluation to tree evaluation (and in this case, the reduction produces a single TEP instance that is guaranteed to correctly encode the circuit, so no enumeration is required).

\begin{theorem}[{\cite[Theorem 4.1 and Lemma 2.3]{Sha25}}]\label{thm:Sha}
    There is a constant $c$ such that the following holds. Let $C$ be a circuit with $S$ gates (described by $O(S\log S)$ bits), let $x \in \zo^{\leq S}$ be an input to $C$, and let $b \in [2, S]$. Then there is a binary $\mathsf{TreeEval}_{h,\ell,2}$ instance with $h \leq c\,S/b$ and $\ell = b$, implicitly presented from $(C, x)$ in space $O(b + \log^2 S)$, a designated bit of whose root value equals $C(x)$. (The reduction of~\cite[Theorem 4.1]{Sha25} presents the instance in space $O(b + \log S)$ from a topologically sorted circuit with its input substituted in; the $O(\log^2 S)$ accounts for bringing an arbitrary pair $(C,x)$ into that form, per~\cite[Lemma 2.3]{Sha25}.)
\end{theorem}

\begin{corollary}\label{cor:circuit}
    For any $\epsilon>0$, size-$S$ circuit evaluation can be decided in $O(\sqrt{S})$ free space, $2^{O(S^\epsilon)}$ catalytic space, and $2^{O(\sqrt{S})}$ time.
\end{corollary}
\begin{proof}
    Apply~\cref{thm:Sha} with $b := \lceil \sqrt S\,\rceil$ and run the algorithm of~\cref{cor:eps}, with parameter $\min(1, 2\epsilon)$, on the implicitly presented instance by emulative composition; output the designated bit of the root value. The free space is $O(b + \log^2 S) + O(\ell + h) = O(\sqrt S)$; the catalytic space is $\exp(O(\ell^{\min(1,2\epsilon)})) = 2^{O(S^{\epsilon})}$; the time is $\poly(2^{\ell+h}) \cdot 2^{O(\sqrt S)} = 2^{O(\sqrt S)}$.
\end{proof}

\section{Statement on AI Use}

\begin{flushleft}
\renewcommand{\arraystretch}{1.3}
\begin{tabular}{@{}l@{\qquad}c@{}}
\textbf{Activity} & \textbf{AI used} \\
Asking the research question & No \\
Coming up with the approach & No \\
Proof development & Yes \\
Writing and exposition & Yes \\
Checking for bugs & Yes \\
Other supporting tasks & Yes \\
\end{tabular}
\end{flushleft}

\begin{sloppypar}
\paragraph{Details (elaborating on the points above).} 
The ideation and writing of the first version of this paper took place without the use of any AI tools. 
We later had the idea of using weak $A$-decoding polynomials (\Cref{def:weakdecpoly}) to make our presentation cleaner and more modular, and used Google Gemini to work out some of the details. 
Later, we asked GPT-6 Astra to simplify this even further, and it suggested the notion of $A$-decoders (\Cref{def:adecoderintro}). Finally, to incorporate these changes and add exposition explaining our techniques, we used GPT-6 Astra and Claude (Opus 5, Opus 5.5, and Fable 5). This included drafting some proofs. We also used these models for proofreading.
\end{sloppypar}

\ifanon
\else
\section{Acknowledgements}
AH and SR thank Vinod Vaikuntanathan for encouraging this work and for helpful discussions. We also thank James Cook, Ian Mertz, Gregory Meyer, and Ryan Williams for useful discussions, and anonymous reviewers for helpful feedback on earlier versions of this paper.
\fi

\bibliographystyle{alpha}
\bibliography{refs}
\appendix

\section{Impossibility of Decoding Polynomials in Characteristic \texorpdfstring{$m$}{m}}
\label{app:nodecoder}

In this appendix, we prove \cref{lem:nodecoder}.

\begin{proof}[Proof of \cref{lem:nodecoder}]
    Let $j^* \in A$ be an element whose additive order in $\Z_m$ is
    $p^a$, for a prime $p$ and an integer $a \geq 1$; note that
    $j^* \neq 0$, since the order of $j^*$ exceeds $1$, and that
    $p \mid m$, since the order of any element of $\Z_m$ divides $m$.

    Suppose toward a contradiction that such $\gamma$ and $P$ exist. Since $j^*$ has additive order $p^a$,
    we have $j^* p^a \equiv 0 \pmod{m}$ and hence
    $\gamma^{j^* p^a} = 1$. Hence, working over $\calR$, we have:
    \begin{align*}
        \frac{m}{p} \cdot (\gamma^{j^*} - 1)^{p^a} &= \frac{m}{p} \cdot (\gamma^{j^* p^a} - 1) && \text{(since $p \mid \binom{p^a}{i}$ for $i \in (0, p^a)$; characteristic $m$)} \\
        &= 0 && \text{($\gamma^{j^* p^a} = 1$)} \\
        \Rightarrow \frac{m}{p} \cdot (P(\gamma^{j^*}) - P(1))^{p^a} &= 0 && \text{($X-1$ divides $P(X)-P(1)$)} \\
        \Rightarrow \frac{m}{p} \cdot (0-1)^{p^a} &= 0 &&\text{(since $j^* \in A \setminus \{0\}$)},
    \end{align*}
    which contradicts the requirement that $\calR$ have characteristic exactly $m$.
\end{proof}

\section{Verifying the Uniformity of the Matching-Vector Family}\label{app:unif}

The matching-vector family of~\cite[Lemma 11]{DGY11} is explicit, but no space bound for generating it is stated there. In this appendix we reproduce the construction and verify that it is log-space uniform, making no claims to originality.

Our notation differs from that of~\cite{DGY11}, whose symbols $h$, $d$, $n$, and $k$ already have other global meanings in this paper. We write $H$ for the ground-set size (their $h$), $w$ for the subset size, $q_i = p_i^{e_i}$ for the prime-power moduli, $\delta = \max_i q_i$ for the degree parameter (their $d$), $N$ for the family size (their $k$), and $d$ for the dimension of the resulting family (their $n$), matching the body's conventions for $N$ and $d$. Recall from~\cref{def:canonicalset} that the canonical set $S^*_m$ contains $0$; the canonical set of~\cite{DGY11} excludes $0$, and accordingly the cross inner products below lie in the zero-free part $S' := S^*_m \setminus \{0\}$.

\begin{lemma}\label{lem:dgyuniform}
    Let $m=\prod_{i=1}^t p_i$ be a product of distinct primes that are given as input. Let $w \geq 1$ and $e_1,\ldots,e_t \geq 0$ be integers such that $\prod_{i=1}^t p_i^{e_i} > w$; write $q_i = p_i^{e_i}$ and $\delta = \max_i q_i$. Let $H \geq w$. Then there exists a space $O(\log N+\delta\log H+\log m)$-uniform\footnote{In both regimes in which we invoke this lemma we have $\log N \geq \delta\log H$, verified in~\cref{sec:mvspecificproof}, so this meets the definition of log-space uniformity.} $S^*_m$-matching-vector family of size $N=\binom{H}{w}$ in $\Z_m^d$, where $d=\sum_{j=0}^{\delta-1}\binom{H}{j}$.
\end{lemma}

\begin{remark}\label{rmk:dgyliteral}
    As literally stated, \cite[Lemma 11]{DGY11} assumes $q_i > w^{1/t}$ for every $i$; inspecting the proof shows that the hypothesis $\prod_i q_i > w$ suffices and that exponents $e_i = 0$ are allowed~\cite[Footnote~2]{GR26}. Since the full proof is reproduced below, we state the relaxed form. Our dimension count $\sum_{j=0}^{\delta-1}\binom{H}{j}$ is also marginally smaller than the $\sum_{j=0}^{\delta}\binom{H}{j}$ stated there, because the polynomials below have degree at most $\delta-1$.
\end{remark}

\newcommand{\cF}{\mathcal{F}}
\newcommand{\cX}{\mathcal{X}}
\newcommand{\cM}{\mathcal{M}}
\newcommand{\mon}{\textsc{mon}}
We first define polynomial matching families. The combinatorial content of the definition is~\cite[Definition 35]{DGY11}; the uniformity clause is ours, and it demands enumeration of both the polynomials and the evaluation points, because the conversion in~\cref{lem:Sud} evaluates monomials at the points.
\begin{definition}[Polynomial matching family; see~{\cite[Definition 35]{DGY11}}]
    Let $S_0 \subseteq \Z_m \setminus \{0\}$. We say that a set of polynomials
    $\cF=\{f_1,\ldots,f_N\}\subseteq \Z_m[z_1,\ldots,z_H]$
    and a set of points $\cX=\{x_1,\ldots,x_N\}\subseteq \Z_m^{H}$
    form a space-$s$ uniform polynomial $S_0$-matching family of size $N$ if
    \begin{itemize}
      \item for all $i\in [N]$, $f_i(x_i)=0$;
      \item for all $i,j\in [N]$ such that $i\neq j$, $f_j(x_i)\in S_0$;
      \item there is a space-$O(s)$ algorithm that, given $i \in [N]$, prints the polynomial $f_i$ (as a list of monomials with their coefficients) and prints the point $x_i$.
    \end{itemize}
\end{definition}

Such a family can be turned into matching vectors by an observation of Sudan.
For $i\in [N]$, let $\supp(f_i)$ denote the set of monomials in the support of the polynomial $f_i$, and define $\supp(\cF)=\bigcup_{i=1}^{N}\supp(f_i)$.
\begin{lemma}[{\cite[Lemma 36]{DGY11}}]\label{lem:Sud}
    Let $\cF,\cX$ be a space-$s$ uniform polynomial $S_0$-matching family of size $N$ over $\Z_m$, and let $\cM \supseteq \supp(\cF)$ be a set of monomials that can be enumerated in space $O(s)$. Then there is a space-$O(s+\log m)$ uniform $(S_0 \cup \{0\})$-matching-vector family of size $N$ in $\Z_m^{d}$, where $d = |\cM|$. In particular, for $S_0 = S' = S^*_m \setminus \{0\}$ this is an $S^*_m$-matching-vector family.
\end{lemma}
\begin{proof}
    Let $\mon_1,\ldots,\mon_d$ be the monomials of $\cM$ in enumeration order, and write
    \[
    f_i =\sum_{l=1}^d c_{il}\cdot \mon_l
    \]
    with $c_{il}\in \Z_m$ (setting $c_{il} = 0$ for monomials outside $\supp(f_i)$). Let $\uv_i\in \Z_m^d$ be the coefficient vector $(c_{i1}, \ldots, c_{id})$, and let $\vecv_j\in \Z_m^d$ be the evaluation vector $(\mon_1(x_j), \ldots, \mon_d(x_j))$. Then $\langle \uv_i, \vecv_j\rangle = f_i(x_j)$, which is $0$ when $i = j$ and lies in $S_0$ when $i \neq j$; so this is an $(S_0 \cup \{0\})$-matching-vector family.

    For uniformity: to print $\uv_i$, enumerate $\cM$ and, for each monomial, print its coefficient in $f_i$ by running the printing algorithm for $f_i$ and matching monomials, in space $O(s)$. To print $\vecv_j$, enumerate $\cM$ and evaluate each monomial at $x_j$: the printing algorithm provides the coordinates of $x_j$, and a monomial is evaluated by multiplying the values of its variables one at a time modulo $m$, in space $O(s+\log m)$.
\end{proof}

We also require low-degree polynomials that detect the Hamming weight modulo a prime power. The statement below is a reformulation of~\cite[Lemma 37]{DGY11}, stated as Theorem 2.16 in Gopalan's thesis~\cite{Gop06}, with its origins in the work of Barrington, Beigel, and Rudich~\cite{Barrington1994}.
\begin{lemma}[cf.~{\cite[Lemma 37]{DGY11}}; {\cite[Theorem 2.16]{Gop06}}]\label{lem:gopalan}
    Let $p$ be a prime, let $e \geq 0$ be an integer, and write $q = p^e$. Let $w \geq 0$ and $H \geq 1$ be integers. There is a symmetric multilinear polynomial $f(z_1,\ldots,z_H)\in \Z_{p}[z_1,\ldots,z_H]$ with $\deg(f)\le q-1$ such that for all $x\in \zo^H$:
    \[
    f(x) \equiv \begin{cases}
        0 \pmod{p} & \text{if } \sum_j x_j \equiv w \pmod{q}, \\
        1 \pmod{p} & \text{otherwise,}
    \end{cases}
    \]
    where $\sum_j x_j$ is computed over $\Z$. Moreover, $f$ can be printed in space $O(q\log (H+1) + \log p)$.\footnote{The explicitness is not stated in the cited sources, but is immediate from the construction; we outline it in the proof.}
\end{lemma}
\begin{proof}[Proof outline]
    For $e = 0$ we have $q = 1$: every weight is congruent to $w$ modulo $1$, and $f = 0$ works. Assume $e \geq 1$. A symmetric multilinear polynomial of degree $\leq q - 1$ has the form $f = \sum_{j=0}^{q-1} c_j E_j(z)$, where $E_j$ denotes the elementary symmetric polynomial of degree $j$ (we reserve the letter $e$ for the exponents). On an input $x \in \zo^H$ of Hamming weight $\sigma$ we have $E_j(x) = \binom{\sigma}{j}$, so $f(x) = \sum_{j < q} c_j \binom{\sigma}{j}$. By Lucas's theorem, for $j < q$ the value $\binom{\sigma}{j} \bmod p$ depends only on the first $e$ base-$p$ digits of $\sigma$, i.e., $\binom{\sigma}{j} \equiv \binom{\sigma \bmod q}{j} \pmod p$. The matrix $\left(\binom{\sigma}{j} \bmod p\right)_{\sigma, j \in \{0, \ldots, q-1\}}$ is lower triangular with unit diagonal, hence invertible over $\Z_p$, so there is a choice of coefficients $(c_j)_{j < q}$ with $\sum_{j<q} c_j \binom{\sigma}{j} \equiv v(\sigma)$ for all $\sigma \in \{0, \ldots, q-1\}$, where $v(\sigma) := \one[\sigma \not\equiv w \pmod q]$; by the periodicity just noted, $f(x) \equiv v(\sigma \bmod q)$ then holds for every weight $\sigma \leq H$. Inverting the binomial transform gives the coefficients explicitly as $c_j = \sum_{\sigma=0}^{j} (-1)^{j-\sigma}\binom{j}{\sigma} v(\sigma) \bmod p$, computable in space $O(\log q + \log p)$ using Lucas's theorem for the binomial coefficients. To print $f$, enumerate the multilinear monomials of degree $< q$ over $z_1, \ldots, z_H$ in space $O(q \log H)$; the coefficient of each monomial depends only on its degree $j$ and equals $c_j$.
\end{proof}

Combining the per-prime polynomials yields weight detection modulo $m$:
\begin{corollary}[cf.~{\cite[Corollary 38]{DGY11}}]\label{cor:poly}
    With $p_1, \ldots, p_t$, $m$, $w$, $q_i$, and $\delta$ as in~\cref{lem:dgyuniform}, there is a space-$O(\delta\log H+\log m)$ algorithm that prints a multilinear polynomial $f(z_1,\ldots,z_H)\in \Z_{m}[z_1,\ldots,z_H]$ with $\deg(f) \leq \delta - 1$ such that for all $x\in \zo^H$ with $\sum_j x_j \leq w$:
    \[
    f(x) \bmod m \in \begin{cases}
        \{0\} & \text{if } \sum_j x_j = w, \\
        S^*_m \setminus \{0\} & \text{if } \sum_j x_j < w.
    \end{cases}
    \]
    In the above, $\sum_j x_j$ is computed over $\Z$.
\end{corollary}
\begin{proof}
    For each $i \in [t]$, let $f_i \in \Z_{p_i}[z_1, \ldots, z_H]$ be the polynomial of~\cref{lem:gopalan} for the prime power $q_i$ and target weight $w$, and let $f$ be the multilinear polynomial over $\Z_m$ obtained by combining the $f_i$ coefficientwise via the Chinese remainder theorem, so that $f \equiv f_i \pmod{p_i}$ for every $i$ and $\deg f \leq \max_i (q_i - 1) \leq \delta - 1$. Fix $x$ of Hamming weight $\sigma \leq w$. Each $f_i(x)$ is $0$ or $1$ mod $p_i$, so $f(x) \bmod m \in S^*_m$. If $\sigma = w$ then every $f_i(x) \equiv 0 \pmod{p_i}$, so $f(x) \equiv 0 \pmod m$. If $\sigma < w$, then since $0 < w - \sigma \leq w < \prod_i q_i$ and the $q_i$ are pairwise coprime, there is an $i$ with $\sigma \not\equiv w \pmod{q_i}$, whence $f_i(x) \equiv 1 \pmod{p_i}$ and $f(x) \not\equiv 0 \pmod m$. For the space bound: each coefficient of $f$ is obtained from the $t$ per-prime coefficients by the Chinese remainder theorem in space $O(\log m)$, and the monomial enumeration costs $O(\delta \log H)$ as in~\cref{lem:gopalan}.
\end{proof}

\begin{claim}\label{clm:map}
    Let $N=\binom{H}{w}$. There is a space $O(\log N + \log H)$-computable bijection from $[N]$ to sets $T\subseteq [H]$ of size $w$.
\end{claim}
\begin{proof}
    It clearly suffices to build a bijection from $[0, N-1]$ to subsets $T \subseteq [0, H-1]$ of size $w$.
    For a given combination $T$, denote the elements of $T$ in decreasing order as $c_w>\ldots>c_1\ge 0$. We express this combination as the number
    \[
    K = \binom{c_w}{w}+\binom{c_{w-1}}{w-1}+\ldots +\binom{c_1}{1}.
    \]
    There is a greedy algorithm that prints $c_w,\ldots,c_1$ given $K$ that runs in space $O(\log N + \log H)$ as follows. We let $S=0$ and $i=w$ and choose $c_i$ maximal such that
    \[
    \binom{c_i}{i}\le K-S \text{ and set } S\la S+\binom{c_i}{i}.
    \]
    We can store $S$ using $O(\log N)$ bits, since all its values are bounded by $N$, and each candidate $c_i$ using $O(\log H)$ bits, so the algorithm runs in space $O(\log N + \log H)$.
\end{proof}

\begin{proof}[Proof of~\cref{lem:dgyuniform}]
    We construct a polynomial matching family and apply~\cref{lem:Sud}.

    We work with subsets $T\subseteq [H]$ of size $w$, indexed by $[N]$ via the space-$O(\log N + \log H)$ bijection of~\cref{clm:map}. Let $f$ be the polynomial of~\cref{cor:poly}. For each such $T$, define $f_T$ as the polynomial obtained from $f$ by setting $z_j=0$ for $j\notin T$, and let $x_T\in \zo^H$ be the indicator vector of $T$.

    This is a polynomial $S'$-matching family: for size-$w$ subsets $T, L$ we have $f_T(x_L) = f(x_{T\cap L})$, and the weight $|T \cap L|$ is at most $w$, with equality if and only if $T = L$. By~\cref{cor:poly}, $f_T(x_T) = 0$ and $f_T(x_L) \in S' = S^*_m \setminus \{0\}$ for $T \neq L$.

    For the uniformity clause: given $i \in [N]$, the algorithm computes $T$ via~\cref{clm:map} and prints $x_T$ directly; it prints $f_T$ by enumerating the multilinear monomials of degree $< \delta$ and printing, for each, the coefficient given by~\cref{cor:poly} if all of the monomial's variables lie in $T$, and $0$ otherwise. This takes space $O(\log N + \delta\log H + \log m)$.

    Finally, apply~\cref{lem:Sud} with $\cM$ the set of all multilinear monomials of degree $< \delta$ over $z_1, \ldots, z_H$, which can be enumerated in space $O(\delta\log H)$. This yields a space-$O(\log N+\delta\log H+\log m)$ uniform $S^*_m$-matching-vector family of size $N$ and dimension $|\cM| = \sum_{j=0}^{\delta-1}\binom{H}{j}$, as claimed.
\end{proof}

\section{Proof of~\ts{\cref{cor:mvspecific}}{Matching Vector Parameters}}\label{sec:mvspecificproof}

The parameter calculation in this appendix is essentially that of~\cite[Corollary 3.4]{GR26}, and we claim no novelty for it: we apply~\cref{thm:mvgeneral} with the parameter choices made in the proof of that corollary and quote its computation of the resulting dimension. What we verify in addition is what the present paper needs beyond existence: the uniformity side condition promised by the footnote of~\cref{lem:dgyuniform} and the final claim of~\cref{thm:mvgeneral}, the exact family size, and how the prime factors are given.

We first set up the dictionary with~\cite{GR26}. Their database size $n$ is our $2^\ell$, so their $\log n$ is our $\ell$; their number of primes, which they call $\ell$, is our $t$; and their parameter is $\eta = \ell^{t-1}/m$. The proof of~\cite[Corollary 3.4]{GR26}, in its second case, which applies once $\log \eta \geq \Omega(t^2)$, instantiates the construction with
\[
\theta = \Theta\!\left(\frac{\log \eta}{t^2}\right), \qquad w \sim \frac{\ell}{\theta}, \qquad H \sim \frac{\ell}{\theta}\cdot 2^{\theta},
\]
so that $N = \binom{H}{w} \geq (H/w)^w = 2^{\ell}$ and $(mw)^{1/t} \sim \ell/(\eta\theta)^{1/t}$, and computes
\begin{equation}\label{eq:grdim}
\sum_{j=0}^{\lfloor (mw)^{1/t}\rfloor} \binom{H}{j} \;\leq\; \exp\!\left(O\!\left(\frac{\ell\,(\log\eta)^{1-1/t}}{t\,\eta^{1/t}}\right)\right).
\end{equation}
Concretely, we take $\theta = \max\!\left(10, \lceil \log\eta/t^2\rceil\right)$, $w = \lceil \ell/\theta \rceil$, and $H = w \cdot 2^{\theta}$, all computable from the input in space $O(\ell + \log m)$; then $N = \binom{H}{w} \geq (H/w)^w = 2^{\theta w} \geq 2^{\ell}$ and $\log N \leq w(\theta + \log e) = O(\ell)$, and the facts quoted above hold for this choice.
By~\cref{thm:mvgeneral} with these choices of $H$ and $w$, we obtain a space-$O(\log N + (mw)^{1/t}\log(H+1) + \log m)$-uniform $S^*_m$-matching-vector family of size $N \geq 2^\ell$ and dimension bounded by~\eqref{eq:grdim}. Restricting the family to its first $2^\ell$ indices under the indexing of~\cref{clm:map} gives a family of size exactly $2^\ell$; a restriction changes neither the matching property nor the uniformity. It remains to translate~\eqref{eq:grdim} and the space bound into the two regimes.

For~\cref{item:grolmusz}, $t \geq 2$ and the primes are fixed, so $\eta = \Theta_m(\ell^{t-1}) \to \infty$ and the case condition holds once $\ell$ exceeds a constant depending on the fixed primes; for smaller $\ell$ the family has constant size and the claims are trivial. We have $\log \eta = \Theta_m(\log \ell)$ and $\eta^{1/t} = \Theta_m(\ell^{1-1/t})$, so the right side of~\eqref{eq:grdim} is $\exp\!\left(O_m\!\left(\ell^{1/t}(\log\ell)^{1-1/t}\right)\right) = \exp\!\left(O_m\!\left(\ell^{1/t}(\log(\ell+2))^{1-1/t}\right)\right)$, as claimed. For the space bound: $\theta = \Theta_m(\log \ell)$, so $\log(H+1) = O(\log w + \theta) = O(\log \ell)$ and $(mw)^{1/t} = \Theta_m\!\left(\ell/(\eta\theta)^{1/t}\right) = \Theta_m\!\left(\ell^{1/t}(\log\ell)^{-1/t}\right)$, whence
\begin{align*}
(mw)^{1/t}\log(H+1) &= O_m\!\left(\ell^{1/t}(\log\ell)^{1-1/t}\right) && \\
&\leq \ell && \text{($\ell$ above a constant depending on the fixed primes)} \\
&\leq \log N. && \text{($N \geq 2^\ell$)}
\end{align*}
Since $\delta \leq (mw)^{1/t}$ (by the choice of exponents in the proof of~\cref{thm:mvgeneral}) and $\log H \leq \log(H+1)$, this also verifies $\log N \geq \delta \log H$, so the family is log-space uniform.

For~\cref{item:manyprimes}, $t = \lceil\sqrt{\log\ell}\,\rceil$ and $p_1, \ldots, p_t$ are the first $t$ odd primes; the uniformity algorithm generates them by trial division in space $O(\log m)$, whereas in~\cref{item:grolmusz} the primes are fixed constants. By the prime number theorem $m = t^{t+o(t)}$, so $\log m = \Theta(\sqrt{\log\ell}\,\log\log\ell)$ and
\begin{align*}
\log \eta &= (t-1)\log\ell - \log m && \text{(definition of $\eta$)} \\
&= \Theta\!\left((\log\ell)^{3/2}\right), && \text{($t = \lceil\sqrt{\log\ell}\,\rceil$; bound on $\log m$)}
\end{align*}
which is $\geq \Omega(t^2) = \Omega(\log \ell)$ for $\ell$ larger than an absolute constant. Then
\[
\log\!\left(\eta^{1/t}\right) = \left(1-\tfrac{1}{t}\right)\log\ell - \frac{\log m}{t} = \log\ell - \Theta\!\left(\sqrt{\log\ell}\,\right),
\qquad\text{i.e.,}\qquad \eta^{1/t} = \ell \cdot 2^{-\Theta(\sqrt{\log\ell})},
\]
so the right side of~\eqref{eq:grdim} is
\[
\exp\!\left(O\!\left(\frac{\ell\,(\log\ell)^{3/2}}{\sqrt{\log\ell}\cdot \ell\cdot 2^{-\Theta(\sqrt{\log\ell})}}\right)\right)
= \exp\!\left(2^{O(\sqrt{\log\ell})}\right)
= \exp\!\left(\exp\!\left(O\!\left(\sqrt{\log\ell}\right)\right)\right),
\]
as claimed. For the space bound: $\theta = \Theta(\log\eta/t^2) = \Theta(\sqrt{\log\ell}\,)$, so $\log(H+1) = O(\log w + \theta) = O(\log\ell)$ and $(mw)^{1/t} = \Theta\!\left(\ell/(\eta\theta)^{1/t}\right) = 2^{O(\sqrt{\log\ell})}$, whence $(mw)^{1/t}\log(H+1) = 2^{O(\sqrt{\log\ell})} \leq \ell \leq \log N$ for $\ell$ larger than an absolute constant; as before this also verifies $\log N \geq \delta\log H$, so the family is log-space uniform.
\qed

\end{document}